\documentclass[10pt,twoside]{article}
\usepackage[utf8]{inputenc}
\usepackage[T1]{fontenc}
\usepackage[margin=0.75in]{geometry}
\usepackage{multicol}
\usepackage{physics}
\usepackage[user,hyperref]{zref}
\usepackage{hyperref}
\usepackage[toc]{appendix}
\usepackage{xcolor}
\usepackage{graphicx} 
\usepackage{cite}
\usepackage{parskip}
\usepackage{newtxtext}
\usepackage{amsmath, amssymb, amsthm}
\usepackage{authblk}

\theoremstyle{plain}      

\newtheorem{proposition}{Proposition}
\newtheorem{lemma}{Lemma}
\newtheorem{corollary}{Corollary}

\theoremstyle{definition}

\theoremstyle{remark}

\usepackage{etoolbox}
\apptocmd{\endproof}{\vspace{2pt}}{}{}

\BeforeBeginEnvironment{theorem}{\vspace{5pt}}
\AfterEndEnvironment{theorem}{\vspace{5pt}}
\BeforeBeginEnvironment{proposition}{\vspace{5pt}}
\AfterEndEnvironment{proposition}{\vspace{5pt}}
\BeforeBeginEnvironment{lemma}{\vspace{5pt}}
\AfterEndEnvironment{lemma}{\vspace{5pt}}
\BeforeBeginEnvironment{corollary}{\vspace{5pt}}
\AfterEndEnvironment{corollary}{\vspace{5pt}}

\makeatletter
\providecommand{\@currentshorttitle}{}
\zref@newprop{shorttitle}{\@currentshorttitle}
\zref@addprop{main}{shorttitle}

\NewDocumentCommand{\labelshort}{om}{%
	\begingroup
	\IfValueT{#1}{%
		\renewcommand{\@currentshorttitle}{#1}%
		\zlabel{#2}%
	}%
	\endgroup
	\label{#2}%
}

\NewDocumentCommand{\nameshortref}{O{}m}{%
	\zref@ifrefundefined{#2}{%
	}{%
		\hyperlink{\zref@extract{#2}{anchor}}{#1\zref@extract{#2}{shorttitle}}%
	}%
}

\usepackage{etoolbox}

\newcommand{\appendixcontent}{}

\newcommand{\toappendix}[1]{%
	\gappto\appendixcontent{#1}%
}

\newcommand{\printappendix}{%
	\appendixcontent
}

\usepackage{cleveref}

\title{\large\bfseries 
	Single Link Removal Perturbation in Szegedy Quantum Walk:\\
	from Graph Completeness Testing to Integrity Monitoring
}
\date{}
\author[1]{Sara Giordano\thanks{\texttt{sgiordan@ucm.es}}}

\author[1,2]{Miguel A. Martin-Delgado\thanks{\texttt{mardel@ucm.es}}}

\affil[1]{Departamento de F\'{\i}sica Te\'orica, Universidad Complutense de Madrid, Madrid, 28040, Spain}
\affil[2]{CCS-Center for Computational Simulation, Campus de Montegancedo Universidad Politecnica de Madrid (UPM), Boadilla del Monte, 28660, Madrid,Spain}

\begin{document}
	\maketitle
	
	%------------------------------------------------------------------------------------
	\begin{abstract}
			We present a rigorous perturbative analysis of the Szegedy quantum walk search algorithm on the complete graph with marked nodes, when a specific anomaly is present in the graph. This is motivated by the problem of monitoring the integrity of dense trusted communication networks with a quantum-assisted procedure. The topology of these networks is modeled as a complete graph, and the anomaly of interest is the disappearance of a single communication link which represents the minimal and spectrally hardest structural defect to detect. Building on the graph-completeness testing algorithm framework, we quantify how the removal of a single unmarked-unmarked edge propagates through the relevant spectral quantities of the Szegedy quantum walk. Denoting by $n$ the total number of nodes of the graph and by $m$ the number of marked nodes, we prove that the perturbation to the transition matrix has spectral norm $\Theta(1/n)$, and that the gap eigenvalue undergoes a strictly negative first-order shift for every $n$ and every number of marked nodes $m$, providing a formal proof of a conjecture from our completeness testing algorithm work; in the regime $m = \Theta(n)$ relevant for the search algorithm, this shift has magnitude $\Theta(1/n^2)$. The corresponding eigenphase shift satisfies $\Delta\theta_\star = \Theta(1/n^2)$ in the same regime. We establish that the rotation angle of the effective subspace under this perturbation is $O(1/n)$ for $m = \Theta(n)$. Finally, we bound the change in success probability to $O(1/\sqrt{n})$ in this same regime of marked nodes, and show that this bound is dominated by the geometric misalignment of the effective subspace rather than by the spectral shift of the eigenphase. These results provide both the theoretical foundations and the fundamental scaling limits of quantum walk-based topology integrity monitoring under minimal structural perturbations.
	\end{abstract}
	\begin{multicols*}{2}
	\section{Introduction}
	
	In this work we introduce a method for monitoring the integrity of dense communication networks, subject to minimal anomalies, grounding on the framework introduced in \cite{GIORDANO2026170305} for testing graph completeness while refining the analysis with a perturbative approach.
	
	Graph completeness testing was originally introduced \cite{GIORDANO2026170305} as a quantum procedure for deciding whether an input arbitrary graph is complete or not by combining the Szegedy quantum walk search \cite{Szegedy_qw, Portugal_book, Aharonov_quantum_walks_on_graphs, Ambainis2007} with quantum phase estimation (QPE) \cite{kitaev1995quantum}. In that framework, the graph structure is encoded in the transition matrix which is embedded in the quantum walk evolution operator, and completeness is inferred from its spectral properties \cite{Portugal_book, Szegedy_qw}. Several practical uses in network structure analysis were already suggested in the seminal work \cite{GIORDANO2026170305}, such as routing algorithms validation or clustering discovery.
	
	This article focuses on a more specific objective.  Our input is a trusted subset of entities in a network which is expected to maintain full pairwise connectivity, so that its topology is modeled by a clique or by a complete graph. Our goal is to detect whether a one link failure occurred, and characterize its effects on the relevant quantities in the Szegedy quantum walk evolution. Moreover, with our analysis we want to highlight the importance and effectiveness of using quantum diagnostic for network integrity monitoring \cite{Centrality_measure, Wang_20, Wang2022, Quad_speed_Apers, Vlasic2025, pagerank_phase, Ortega_complex_phase, Paparo2012, Feldman2010, romeo2025probinggraphtopologylocal, Reitzner2009}, precisely defining its capabilities in this specific but crucial setting. Perturbative approaches to quantum walks have also been used to study robustness to generic transition-matrix errors and their propagation to hitting times \cite{Chiang2010, Chiang2013}; our work instead considers theoretically motivated perturbation and extends the analysis to subspace rotation and search success probability.
	
	The situation of having one missing link is particularly interesting because the most difficult incomplete graph to distinguish from a complete one is precisely the one obtained by deleting only one edge. This is because we can expect its spectral properties to be close to the complete case. Therefore, the single-edge-removal setting is the minimally different and spectrally hardest case of incomplete graph. For this reason, it provides a natural bridge between the aforementioned graph-completeness algorithm and a more refined perturbative analysis for a single edge deletion. Single-edge-removal perturbations of quantum walks have also been studied numerically in other graph families, e.g. for embedded hypercubes, where the removal of a single edge was found not to destroy the exponential quantum speedup in hitting time~\cite{Makmal2016}. Here we instead provide an analytical treatment of a single edge removal on the complete graph with marked nodes, quantifying its effect on the search algorithm rather than on the hitting time alone.
	
	\section{Motivation and application scenario}\label{sec:motivation_and_application}
	Consider a monitored infrastructure composed of a set of critical network entities, such as servers, services, control-plane nodes,	orchestration agents, or protected zones. We model their communication pattern by an undirected graph
	\begin{equation}
		G^0 = (V, \mathcal{E}^0),
	\end{equation}
	where vertices represent the monitored entities and edges authorized bidirectional communication links.
	
	In many operational settings, the relevant subnetwork is designed to be highly connected. Examples include tightly coupled data-center fabrics~\cite{lacin2026, guo2009bcube}, east--west communication layers inside protected service domains~\cite{zhuang2024microsegmented}, coordination networks in which every pair of nodes in the control-plane is expected to	maintain direct pairwise communication~\cite{jeffery2021kubernetes}, and secure control-plane overlays~\cite{lahmadi2013disco}. In the idealized regime adopted here, this trusted subgraph of communicating nodes is modeled as a complete graph $K_n$.
	
	The integrity monitoring problem then becomes the following: determine whether the currently observed communication graph still	matches the expected dense communication pattern, or whether it has
	suffered a structural anomaly. The anomaly of interest is the disappearance of a single expected link,
	\begin{equation}
		\tilde{G} = (V,\; \mathcal{E}^0 \setminus \{(u,v)\}),
	\end{equation}
	for some unknown pair $u, v \in V$, $u \neq v$. This missing edge may correspond to a physical link failure, a routing malfunction, a misconfiguration, or a cyberattack that selectively disrupts the link communication.
	
	This single edge deletion is a realistic model for early-stage network failures and stealthy topology deviations, since in many real systems the first visible symptom of degradation is not a complete network collapse, but the loss of connectivity between a very small number of node pairs.
	
	We clarify that the monitoring task we analyze consists in deciding whether the observed communication graph remains topologically consistent with the trusted complete graph, and in quantifying the changes in the relevant quantities of the Szegedy quantum walk evolution. This method does not localize the missing edge, which may be pursued as a subsequent step once an anomaly has been detected, and it is not meant to replace general-purpose intrusion detection systems based on traffic analysis. It provides a specialized mechanism for detecting subtle structural anomalies in environments where the expected connectivity pattern is known in advance and small deviations are operationally significant.
	
	We begin from the complete graph with marked nodes and study the effect of removing one unmarked--unmarked edge, that is, an edge $\{u,v\} \subseteq V \setminus M$, where $M \subseteq V$ is the designated marked set. This constraint is essential: the perturbative results characterize the effect of structural changes in the unmarked sector on the walk dynamics, and the analysis does not extend to the removal of edges incident to marked nodes. 
		
	This perturbation can be represented as a low-rank modification of the transition matrix, which preserves column-stochasticity through an appropriate redistribution of probability weight. We estimate the shift of the relevant simple eigenvalue, used as a proxy of completeness in the completeness testing algorithm in \cite{GIORDANO2026170305}, the induced phase variation, the rotation of
	the effective invariant subspace, and the corresponding change in the probability of measuring a marked node.
		
	These technical results quantify the sensitivity of the quantum monitoring procedure to the smallest possible structural change in the edge structure for a complete graph.
	
	\section{Szegedy Quantum Walk Construction}\label{sec:szegedy_quantum_walk_construction}
	Szegedy's quantum walk corresponds to the quantization of Reversible Markov chains \cite{Szegedy_qw,Portugal_book, santos2016szegedysquantumwalkqueries, Ambainis2005} and provides a systematic way to define a discrete-time quantum walk from a classical transition matrix $P$, usually represented by a graph $G$ \cite{Shenvi2003}. The unitary evolution operator of the quantum walk $W$ encodes the spectral information of the discriminant matrix $C$ which depends on the transition matrix $P$ and is defined as $C_{u,v}=\sqrt{P_{v,u}P_{u,v}}$. In particular, when marked nodes are present on the graph, the evolution operator spectrum plays a central role in quantum search algorithms based on Szegedy quantum walk. 
	
	In this note, we study the effects of removing edges from the underlying graph on the spectrum of the unitary operator $W$, when the latter is a complete graph $K_n$  with marked nodes, focusing on the removal of a single edge.
	
	Let $P^0$ be the column-stochastic transition matrix of an ergodic Markov chain on a graph $G^0=(V,\mathcal{E}^0)$ with $|V|=n$. Define the Hilbert space
	\begin{equation}
		\mathcal{H} = \mathrm{span}\{\, \ket{u}\ket{v} : u,v \in V \,\},
	\end{equation}
	where $\ket{v},\ket{u} \in \mathbb{C}^n$. For each $u\in V$, define the state
	\begin{equation}
		\ket{\psi_u} = \ket{u} \otimes \sum_{v \in V} \sqrt{P^0_{v,u}}\,\ket{v}.
	\end{equation}
	Similarly, for each $v\in V$,
	\begin{equation}
		\ket{\phi_v} = \sum_{u \in V} \sqrt{P^0_{v,u}}\,\ket{u} \otimes \ket{v}.
	\end{equation}
	
	We define the isometries
	\begin{equation}
		A = \sum_{u \in V}|\psi_u\rangle\langle u|,
		\qquad
		B = \sum_{v \in V}|\phi_v\rangle\langle v|,
	\end{equation}
	which satisfy $\|A\|_2 = \|B\|_2 = 1$ since $\langle\psi_u|\psi_{u'}\rangle
	= \delta_{uu'}$ and $\langle\phi_v|\phi_{v'}\rangle = \delta_{vv'}$,
	the associated subspaces are:
	\begin{equation}
		\mathcal{A} = \mathrm{span}\{\ket{\psi_u} : u \in V\}
		\qquad
		\mathcal{B} = \mathrm{span}\{\ket{\phi_v} : v \in V\}.
	\end{equation}
	
	Defining the reflections
	\begin{equation}
		R_A = 2 \sum_{u} \ket{\psi_u} \bra{\psi_u} - I, \quad R_B = 2 \sum_{v} \ket{\phi_v} \bra{\phi_v} - I,
	\end{equation}
	which reflects over the subspaces $\mathcal{A}$ and $\mathcal{B}$, we can write finally the Szegedy walk operator:
	\begin{equation}\label{eq:quantum_walk_operator}
		W^0 = R_B R_A,
	\end{equation}
	
	The spectrum of $W^0$ is determined by the discriminant matrix $C^0$ with entries
	\begin{equation}
		C^0_{u,v} = \sqrt{P^0_{v,u} P^0_{u,v}}.
	\end{equation}
	
	If $\lambda$ is an eigenvalue of $C^0$, then $W^0$ has corresponding eigenvalues $e^{\pm i \arccos(\lambda)}$.
	
	\subsection{Transition matrix of the complete graph with marked nodes}
	Initially, the underlying graph is the complete graph $K_n$. The transition matrix $P^0$ is column-stochastic:
	\begin{equation}
		P^0_{v,u} = \begin{cases}
			\frac{1}{n-1}, & v \neq u, \\
			0, & v=u.
		\end{cases}
	\end{equation}
	Marking a set $M$ of $m= |M|$ vertices is modeled by making the corresponding columns in $P^0$ absorbing or "sinks", obtaining the matrix $P$. Calling $\mathbf{e}_j$ the column vector of the computational basis, for which $\mathbf{e_j^i}=0 $ for each $i\neq j$ and $\mathbf{e^{ j}_j}=1$, we mark the elements on the graph by writing the columns of $P$ as follows:
	\begin{equation}
		P^0_{\cdot,j}\to P_{\cdot, j} = \mathbf{e}_j, \quad j \in M,
	\end{equation}
	where $P^0_{\cdot,j}$ and $P_{\cdot,j}$ are respectively the columns of $P^0$ and $P$.
	
	At the level of the underlying graph it means that for the marked nodes there are no more outgoing edges, only ingoing ones and a self loop is added. The evolution operator based on this marked transition matrix $P$ allows for the development of quantum walk search algorithms, which hopefully maximize the probability of finding the quantum walker on the marked nodes.

	By removing an edge of the graph, we are applying a perturbation to the corresponding transition matrix. Indeed the change we are applying can be written as a low rank matrix and its norm is much smaller compared to that of the transition matrix. To assess that the applied change is a perturbation we will compare the spectral norms of the matrix, $\|P\|_2$ and of the perturbation matrix.
	
	Because Szegedy's walk spectrum depends on the discriminant matrix $C$, we evaluate the perturbation on $P$ and then track how this affects $C$ and consequently the quantum walk operator $W$.
	
	\subsection{Discriminant with marked vertices: block form and exact spectrum}\label{subsec:discriminant_with_marked_vertices_block_form} 
	The discriminant $C^0$ (Portugal's notation) is
	\begin{equation}
		C^0_{u,v}=\sqrt{\,P^0_{v,u}\,P^0_{u,v}\,},
	\end{equation}
	with $P^0$ transition matrix without marked nodes. Analogously, we can define $C$ based on the matrix $P$ with marked nodes:
		\begin{equation}
		C_{u,v}=\sqrt{\,P_{v,u}\,P_{u,v}\,}.
	\end{equation}
	By construction, $C$ is real and symmetric. Recalling that we are starting with a complete graph case, the $C$ matrix has a block structure 
	\begin{equation}
		C \;=\; I_{(m)}\ \oplus\ \frac{1}{n-1}\left(J_{(k)} - I_{(k)}\right),
	\end{equation}
	where $k=n-m$, $J_{(k)}$ is the $k\times k$ all-ones matrix and $I_{(m)}$ is the $m\times m $ identity matrix. Consequently, the eigenvalues $\lambda$ of $C$ are
	\begin{gather}\label{eq:eigenvalues_C}\notag
		\underbrace{\lambda_M=1,}_{\text{multiplicity } m} \; \lambda_\star=\frac{k-1}{n-1}=\underbrace{\frac{n-m-1}{n-1}}_{\text{simple eigenvalue}}\\, 
		\underbrace{\lambda_B=-\frac{1}{n-1}} _{\text{multiplicity }k-1}.
	\end{gather}
	These correspond to three invariant subspaces respectively $\mathcal{U}_M,\mathcal{U}_\star,\mathcal{U}_B$. The $m$ eigenvalues $1$ derive from the absorbing marked block of $C$ which comes from the marked part of the transition matrix $P$; the remaining spectrum comes from the unmarked block, an invariant subspace of dimension $1$ and the bulk subspace of dimension $k-1=n-m-1$.
	
	The eigenspaces $\mathcal{U}_M$, $\mathcal{U}_\star$, $\mathcal{U}_B$ of $C$ are subspaces of $\mathbb{R}^n$. By the spectral correspondence between $C$ and $W$~\cite{Szegedy_qw, Portugal_book}, each eigenspace $\mathcal{U}$ of $C$ with eigenvalue $\lambda$ corresponds to an invariant subspace	$\mathcal{V}$ of $W$ in $\mathcal{H}$, spanned by the eigenvectors
	$\ket{\omega^\pm}$ with eigenvalues $e^{\pm 2i\arccos(\lambda)}$. We denote by $\mathcal{V}_M$, $\mathcal{V}_\star$, $\mathcal{V}_B$, and $\mathcal{V}_{\mathrm{eff}} = \mathcal{V}_M \oplus \mathcal{V}_\star$ the corresponding invariant subspaces of $W$ in $\mathcal{H}$.
	
	The full spectral decomposition of $W$ also contains a large $(+1)$-eigenspace $\mathcal{A}^\perp \cap \mathcal{B}^\perp$ of dimension $n^2 - 2n + m$, orthogonal to $\mathcal{V}_{\mathrm{eff}} \oplus \mathcal{V}_B$, so that
	\begin{equation}
		\mathcal{H} = \mathcal{V}_{\mathrm{eff}} \oplus \mathcal{V}_B
		\oplus (\mathcal{A}^\perp \cap \mathcal{B}^\perp).
	\end{equation}
	
	The initial state $\ket{\psi_0} = \frac{1}{\sqrt{n}}\sum_{x \in V} \ket{\psi_x}$, defined as uniformly distributed over all the nodes of the unmarked complete graph, satisfies $\ket{\psi_0} \perp \mathcal{V}_B$ exactly,	as shown in~\cite{Portugal_book}. In other words $\ket{\psi_0} \in \mathcal{V}_{M}\oplus \mathcal{V}_\star \oplus \mathcal{A}^\perp \cap \mathcal{B}^\perp = \mathcal{V}_{\mathrm{eff}}\oplus \mathcal{A}^\perp \cap \mathcal{B}^\perp $. Its remaining component outside $\mathcal{V}_{\mathrm{eff}}$ lies in $\mathcal{A}^\perp \cap \mathcal{B}^\perp$, which is a $(+1)$-eigenspace of $W$ satisfying $\Pi_M(\mathcal{A}^\perp \cap \mathcal{B}^\perp) = 0$~\cite{Portugal_book}. Since $W^t$ acts as the identity on $\mathcal{A}^\perp \cap	\mathcal{B}^\perp$, this has direct consequences for the probability of finding the walker on a marked node, which we will use in \Cref{sec:effects_on_succ_prob}:
	\begin{align}
		\Pi_M \ket{\psi_0} &= \Pi_M \Pi_{\mathcal{V}_{\mathrm{eff}}}\ket{\psi_0},\\
		\Pi_M W^t \ket{\psi_0} &= \Pi_M W^t
		\Pi_{\mathcal{V}_{\mathrm{eff}}}\ket{\psi_0}.
	\end{align}
	
	Among the three subspaces $\mathcal{V}_M$, $\mathcal{V}_\star$,	$\mathcal{V}_B$, only $\mathcal{V}_{\mathrm{eff}} = \mathcal{V}_M \oplus \mathcal{V}_\star$ is relevant for the success probability	$p_M(t)$. This follows from $\ket{\psi_0} \perp \mathcal{V}_B$, so the bulk subspace plays no role in the dynamics, and the component of $\ket{\psi_0}$	outside $\mathcal{V}_{\mathrm{eff}}$ lies in $\mathcal{A}^\perp \cap \mathcal{B}^\perp$, which is invisible to $\Pi_M$. The effective subspace $\mathcal{V}_{\mathrm{eff}}$ therefore includes all the spectral component of $|\psi_0\rangle$ that is relevant for detecting the marked nodes.
	
	It is crucial to remark that the spectral norm of a matrix $F$ is given by:
	\begin{equation}
		\|F\|_2=\sqrt{\lambda_{\mathrm{max}}(D)},
	\end{equation} 
	 (see e.g.\cite{HornJohnson2012}, Example 5.6.6, Sec. 5.6) where $\lambda_{\mathrm{max}}(D)$ indicates the maximum eigenvalue of the discriminant matrix $D=F^*F$ of $F$, i.e. the maximum singular value of $F$. This means that the spectral norm of $P$ is equal to the maximum eigenvalue of $C$, $\lambda_M=1$. Once we will have formalized the perturbation matrix, we will compare its spectral norm with $\|P\|_2=1$. This comparison will confirm that the removal of a single edge is indeed a perturbation for the original transition matrix. 
	
	\section{Single Edge Removal Perturbation}
	
	In the complete graph $K_n$ with a marked set $M$ of $m$ vertices, there are three classes of edges:
	\begin{itemize}
		\item Marked--marked: connecting vertices $a - b$ with $a,b\in M$. After marking, the columns corresponding to $a$ and $b$ are absorbing ($P_{\cdot,a}=e_a$ and $P_{\cdot,b}=e_b$), so all outgoing edges from $a$ and $b$ are removed, including $a\to b$ and $b \to a$.
		\item Marked--unmarked: connecting vertices $a - u$ with $a\in M$, $u\in V\setminus M$. While also $a \to u$ is removed for absorption, removing the $u \to a$ edge causes a rank-one perturbation.
		\item Unmarked--unmarked: connecting vertices $u - v$ with $u,v\in V \setminus M$, $u\neq v$. This case is considered nontrivial after marking, it involves the removal of $u \to v$  and $v \to u$.
	\end{itemize}
	We treat the last of these three cases. We will treat one undirected edge removal as two directed-edge removals, each one being a rank-1 perturbation to the transition matrix.

	\subsection{Perturbation representation}
	A rank-1 perturbation to the matrix $P$ consists in adding or subtracting a matrix of rank-1 to the matrix $P$.
	Since it should remain a graph transition matrix, we have to guarantee its column-stochasticity. Thus, when we remove edge $u-v$ by subtracting the probabilities $p = P_{v,u}=P_{u,v}$ from the columns $P[:,u]$ and $P[v,:]$, we renormalize these columns by redistributing the weight $p$ among the remaining columns elements. Thus, when removing one edge only two columns $P[:,u]$ and $P[v,:]$ changes, with each change consisting of a rank-1 perturbation of the transition matrix. Considering a complete graph with uniform transition probability, we have $p=\frac{1}{n-1}$. In this section we will formalize this perturbation in terms of a matrix operation.

	Let $\mathbf{c}$ denote the original column $P[:,u]$, while $\mathbf{\tilde{c}}$ denotes the column after the edge removal and $\mathbf{c}_i$ denotes a column element. We also call $\mathbf{r}^{(u)}$ the column vector used to distribute the weight $p$ among the other column elements. We remove from $\mathbf{c}$ the weight $p=P_{v,u}=\mathbf{c}_v$ and redistribute it with the column vector $\mathbf{r}^{(u)}$, where $\sum_i \mathbf{r}^{(u)}_i=1$ and $\mathbf{r}^{(u)}_v=0$. This reads:
	\begin{equation}
		\mathbf{\tilde{c}} = \mathbf{c} - p \mathbf{e}_v + p \mathbf{r}^{(u)} = \mathbf{c} + p(\mathbf{r}^{(u)}-\mathbf{e}_v).
	\end{equation}
	Thus the matrix representing the perturbation induced from the removal of one directed edge is
	\begin{equation}\label{eq:perturbation_direct_edge}
		\Delta_u^{\mathrm{dir}} = (p(\mathbf{r}^{(u)}-\mathbf{e}_v)) \mathbf{e}_u^T,
	\end{equation}
	a rank-one matrix. When we remove both $u \to v$ and $v \to u$ we have:
	\begin{equation}\label{eq:perturbation_first_expr}
		\Delta = \Delta_u^{\mathrm{dir}}+\Delta_v^{\mathrm{dir}}= (p( \mathbf{r}^{(u)}-\mathbf{e}_v)) \mathbf{e}_u^T + (p(\mathbf{r}^{(v)}-\mathbf{e}_u)) \mathbf{e}_v^T
	\end{equation}
	a rank-2 matrix, where $\sum_i \mathbf{r}^{(v)}_i=1$ and $\mathbf{r}^{(v)}_u=0$.
	
	Let us define
	\begin{equation}
		\mathbf{a}=\mathbf{r}^{(u)}-\mathbf{e}_v,\qquad \mathbf{b}=\mathbf{r}^{(v)}-\mathbf{e}_u,
	\end{equation}
	
	so that the perturbation in Eq.\eqref{eq:perturbation_first_expr} can be written as	\begin{equation}\label{eq:perturbation}
		\Delta \;=\; p\left(\mathbf{a}\,\mathbf{e}_u^{\!T} \,+\, \mathbf{b}\,\mathbf{e}_v^{\!T}\right).
	\end{equation}

	\paragraph{Norm of the perturbation.}
	Since we are interested in the effect of the perturbation $\Delta$ on the spectrum of the evolution operator $W$ we must calculate its spectral norm $\|\Delta\|_2$. We use the Frobenius norm \cite{HornJohnson2012} which provides an upper and lower bound to the spectral norm, the details on how to derive the following proposition are reported in \nameshortref{app:perturbation_norm_approx}.
	
	\begin{proposition}[Spectral norm of the perturbation]\label{prop:perturbation_norm}
		Let $\Delta$ be the perturbation of the transition matrix induced by the removal of the undirected edge $\{u,v\}$ from a complete graph with marked nodes $K_n$, with $p = \frac{1}{n-1}$ and weight distributed uniformly over the $n-2$ remaining neighbours of each node. Then the spectral norm of $\Delta$ satisfies
		\begin{equation}
			\frac{1}{n-1}\sqrt{1+\frac{1}{n-2}}	\;\leq\; \|\Delta\|_2 \;\leq\; \frac{1}{n-1}\sqrt{2+\frac{2}{n-2}},
		\end{equation}
		and in particular for large $n$
		\begin{equation}
			\|\Delta\|_2 = \Theta\!\left(\frac{1}{n}\right).
		\end{equation}
	\end{proposition}
	Notice that in this proposition we did not set any regime for the number of marked nodes $m$.
	\newcounter{AppCount}

	%FIRST APPENDIX BEGIN----------------------------------------------------------------------------------------------------------------------------------------
		\toappendix{%
			\setcounter{AppCount}{1}
			\subsection*{\Alph{AppCount}. \space Norm of the perturbation $\Delta$}\labelshort[\Alph{AppCount}]{app:perturbation_norm_approx}
			\begin{proof}[Proof of \Cref{prop:perturbation_norm}]
				\medskip\noindent\textbf{Spectral norm of the perturbation}
				
				\medskip\medskip\noindent\textit{\textbf{Upper bound through Frobenius norm.}}
				The Frobenius norm $\|\Delta\|_F$ satisfies $\|\Delta\|_F^2=\text{tr}(\Delta(\Delta)^{\!T})$. Using $\mathbf{e}_u^{\!T}\mathbf{e}_v=0$ we get
				\begin{equation}
					\Delta(\Delta)^{\!T}
					= p^2\left( \mathbf{a a}^{\!T} + \mathbf{b b}^{\!T}\right),
				\end{equation}
				hence
				\begin{equation}
					\|\Delta\|_F^2
					= \mathrm{tr}\!\left(\Delta(\Delta)^{\!T}\right)
					= p^2\left(\|\mathbf{a}\|_2^2 + \|\mathbf{b}\|_2^2\right),
				\end{equation}
				where we used the property $\mathrm{tr}(\mathbf{aa}^{\!T})=\mathbf{a}^{\!T}\mathbf{a}$ for both $\mathbf{a}$ and $\mathbf{b}$. We substitute $\mathbf{a}=\mathbf{r}^{(u)}-\mathbf{e}_v$ and $\mathbf{b}=\mathbf{r}^{(v)}-\mathbf{e}_u$ and we obtain:
				\begin{equation}
					\;\|\Delta\|_F^2
					= p^2\left(\|\mathbf{r}^{(u)}\|_2^2+\|\mathbf{r}^{(v)}\|_2^2 + 2\right),\;
				\end{equation}
				because $\|\mathbf{e}_u\|^2_2=\|\mathbf{e}_v\|^2_2=1$ and because $\mathbf{r}^{(u)}_v=\mathbf{r}^{(v)}_u=0$, thus the cross products between $\mathbf{r}^{(u)}/\mathbf{r}^{(v)}$ and $\mathbf{e}_u/ \mathbf{e}_v$ are zero.
				
				If $\mathbf{r}^{(u)}$ and $\mathbf{r}^{(v)}$ distribute respectively the weights $p$ uniformly over a number of nodes $K_u$ and $K_v$, then
				$\|\mathbf{r}^{(u)}\|_2^2=\sum(\frac{1}{K_u})^2=\frac{K_u}{K_u^2}=\frac{1}{K_u}$ and $\|\mathbf{r}^{(v)}\|_2^2=\frac{1}{K_v}$, so
				\begin{equation}
					\;\|\Delta\|_F
					= p\,\sqrt{\,2+\frac{1}{K_u}+\frac{1}{K_v}\,}\;. \;
				\end{equation}
				Considering that $K_u=K_v=n-2$ when the weight is distributed uniformly over the other nodes connected to $u$ and $v$, and that $p=\frac{1}{n-1}$ then
				\begin{equation}
					\;\|\Delta\|_F=\frac{1}{n-1}\sqrt{2+ \frac{2}{n-2}}= \;\Theta\!\left(\frac{1}{n}\right)\;.
				\end{equation}
				This provides an upper bound on the spectral norm since 
				\begin{equation}
					\|\Delta\|_2\! \le \! \|\Delta\|_F=\Theta\!\left(\frac{1}{n}\right).
				\end{equation}
				\medskip\noindent\textit{\textbf{Lower bound through Frobenius norm.}}  
				Let the rank of the perturbation be $r=\mathrm{rank}(\Delta)$. Since between the Frobenius norm and the spectral norm it holds that\cite{HornJohnson2012}:
				\begin{equation}
					\|\Delta\|_2\geq \frac{\|\Delta\|_F}{\sqrt{r}}\;,
				\end{equation}
				it follows immediately that
				\begin{equation}
					\frac{\|\Delta\|_F}{\sqrt{r}} \;\le\; \|\Delta\|_2 \;\le\; \|\Delta\|_F,
				\end{equation}
				and since $r=\text{rank}(\Delta)=2$ in our case, we have:
				\begin{equation}
					\frac{\|\Delta\|_F}{\sqrt{2}} \;\le\; \|\Delta\|_2 \;\le\; \|\Delta\|_F.
				\end{equation}\qedhere
			\end{proof}
		}
	%FIRST APPENDIX END----------------------------------------------------------------------------------------------------------------------------------------

	More specifically, the Frobenius norm $\|\Delta\|_F$ is
	\begin{equation}
		\|\Delta\|_F=\frac{1}{n-1}\sqrt{2+ \frac{2}{n-2}}= O\!\left(\frac{1}{n}\right),
	\end{equation}
	and 
	\begin{equation}
		\frac{\|\Delta\|_F}{\sqrt{r}}\leq \|\Delta\|_2  \leq \|\Delta\|_F
	\end{equation}
	with $r$ the rank of the perturbation. Thus, whenever $\Delta$ has fixed low rank (e.g.\ rank one or two), the Frobenius norm provides a control of the spectral norm both from above and below. In our case $r=2$ and thus $\|\Delta\|_2 = \Theta\!\left(\frac{1}{n}\right)$.

	At the end of \Cref{subsec:discriminant_with_marked_vertices_block_form} we calculated $\|P\|_2=1$, comparing it with $\|\Delta\|_2=O(\frac{1}{n})$ we can assume that for large $n$ we have $\|P\|_2 \gg \|\Delta\|_2$. This confirms that $\Delta$ is a perturbation for $P$.

	\subsection{Effects of the perturbation on the spectrum of the Szegedy discriminant matrix}\label{subsec:effects_of_the_perturbation_on_the_discriminant_matrix}
	Now that we have an evaluation of the spectral norm of the perturbation for the matrix $P$ we can analyze its discriminant matrix $C$ whose eigenvalues are involved in the quantum walk evolution given by the operator $W$. We need to evaluate the effects of the perturbation $\Delta$ on the spectra of $C$. In particular, we want to evaluate how much the eigenvalues of $C$ change due to the perturbation and if they increase or decrease in value. We call $E=\tilde{C}-C$ the perturbation matrix acting on the discriminant.
	
	Since $C$ and $\tilde{C}$ are real symmetric, Weyl's inequality \cite{Bhatia1997} (Chapter III, Sec.~2) gives
	\begin{equation}\label{eq:bauer_fike_C}
		\max_i |\lambda_i(\tilde{C})-\lambda_i(C)| \;\le\; \|E\|_2.
	\end{equation}
	where $\lambda_i(\tilde{C})$ and $\lambda_i(C)$ are the eigenvalues respectively of $\tilde{C}$ and $C$ (see Eq.\eqref{eq:eigenvalues_C}). This type of bound on the spectral gap of a quantized transition matrix under a generic perturbation was first applied to Szegedy quantum walks by Chiang~\cite{Chiang2010}, and extended to hitting-time bounds by Chiang and Gomez~\cite{Chiang2013}; here we specialize this approach to the structured perturbation $E$. Therefore, we need to calculate $\|E\|_2$. Here we report the results detailed in \nameshortref{app:perturbation_discriminant}, where we use the definition of $\Delta$ to derive the explicit form of $E$ and its norm $\|E\|_2$.
	
	\begin{proposition}[Structure and norm of $E = \tilde{C} - C$]
		\label{prop:E_structure}
		Let $E = \tilde{C} - C$ be the perturbation of the discriminant matrix induced by the removal of the edge $\{u,v\}$ from $K_n$, with $p = \frac{1}{n-1}$ and redistribution weights $\mathbf{r}^{(u)}_x = \mathbf{r}^{(v)}_x = \frac{1}{n-2}$ for all $x \neq u,v$.
		\begin{enumerate}
			\item[\textup{(i)}] \textbf{Support and explicit entries.}
			$E$ is symmetric, supported only on rows and columns $\{u,v\}$,	and of rank at most $2$. Its non-zero entries are:
			\begin{align}
				E_{u,v} &= E_{v,u} = -\frac{1}{n-1}, \\[4pt] \notag
				E_{x,u} &= E_{x,v} = \frac{1}{2(n-1)}\,\mathbf{r}^{(u)}_x
				+ O\!\left(\frac{(\mathbf{r}^{(u)}_x)^2}{n}\right)
				\\
				& \approx \frac{1}{n^2} + O\!\left(\frac{1}{n^3}\right),
				\quad x \neq u,v.
			\end{align}
			In particular, $E_{u,u} = E_{v,v} = 0$, and $E_{x,y} = 0$ whenever $x,y \notin \{u,v\}$ or $x$ or $y$ is a marked node.
			
			\item[\textup{(ii)}] \textbf{Spectral norm.}
			The spectral norm of $E$ satisfies
			\begin{equation}
				\|E\|_2 = \Theta\!\left(\frac{1}{n}\right),
			\end{equation}
			with leading contribution from $E_{u,v} = -p = -\frac{1}{n-1}$.
		\end{enumerate}
	\end{proposition}

	%SECOND APPENDIX BEGIN----------------------------------------------------------------------------------------------------------------------------------------
		\toappendix{
			\stepcounter{AppCount}
			\subsection*{\Alph{AppCount}. \space Perturbation $E$ to the discriminant matrix $C$}\labelshort[\Alph{AppCount}]{app:perturbation_discriminant}
			In this appendix we derive the explicit form of the perturbation $E$ calculating its value at the first order approximation, proving \Cref{prop:E_structure} (i) and (ii).
			
			\begin{proof}[Proof of \Cref{prop:E_structure} (i)]
				\medskip\noindent\textit{\textbf{Structure and norm of $E = \tilde{C} - C$: support and explicit entries.}}
								
				Recalling that matrix $C$ is $C_{x,y}=\sqrt{P_{y,x}P_{x,y}}$, since $E=\tilde{C}-C$ we have:
				\begin{equation}
					E_{x,y} \;=\; \sqrt{(P_{y,x}+\Delta_{y,x})(P_{x,y}+\Delta_{x,y})}\;-\;\sqrt{P_{y,x}P_{x,y}}.
				\end{equation}
				
				Because $\Delta$ modifies only the columns corresponding to the two nodes of the removed edge, $u$ and $v$, the matrix $E$ has support only on rows and columns $u$ and $v$. More precisely, recalling that $E$ is symmetric, for any $x$ we have:
				\begin{gather}\notag
					E_{x,u}=\sqrt{(P_{u,x}+\Delta_{u,x})(P_{x,u}+\Delta_{x,u})}- \sqrt{P_{u,x}P_{x,u}},\\
					E_{x,v}=\sqrt{(P_{v,x}+\Delta_{v,x})(P_{x,v}+\Delta_{x,v})}- \sqrt{P_{v,x}P_{x,v}}.
				\end{gather}
				Considering that the perturbation matrix has only columns $u$ and $v$ different from zero, then $\Delta_{u,x}=\Delta_{v,x}=0 $ for each $x$, and thus we can write
				\begin{gather}\notag
					E_{x,u} = \sqrt{P_{u,x}}\left(\sqrt{P_{x,u}+\Delta_{x,u}}-\sqrt{P_{x,u}}\right),\\
					E_{x,v} = \sqrt{P_{v,x}}\left(\sqrt{P_{x,v}+\Delta_{x,v}}-\sqrt{P_{x,v}}\right).
				\end{gather}
				All other entries with $x,y$ different from $u$ and $v$ are zero, and $E_{x,y}=0$ also when $x$ or $y$ corresponds to a marked node. Recalling the perturbation in Eq.\eqref{eq:perturbation_direct_edge} and writing both of the directed edge perturbations elementwise, we obtain
				\begin{equation}
					\Delta^{dir}_{x,u}=p\left(\mathbf{r}^{(u)}_x-\delta_{xv}\right), \quad \Delta^{dir}_{x,v}=p\left(\mathbf{r}^{(v)}_x-\delta_{xu}\right).
				\end{equation}
				Since on the marked complete graph we have $P_{x,u}=P_{u,x}=P_{x,v}=P_{v,x}=p=\tfrac{1}{n-1}$, so:
				\begin{equation}\label{eq:E_non_zero_general}
						E_{x,u} \;=\; \frac{1}{n-1}\left(\sqrt{\,1+\mathbf{r}^{(u)}_x - \delta_{x,v}\,}-1\right),\quad
						E_{x,v} \;=\; \frac{1}{n-1}\left(\sqrt{\,1+\mathbf{r}^{(v)}_x - \delta_{x,u}\,}-1\right)\;.
				\end{equation}
				
				\medskip\noindent\textbf{\textit{First order approximation.}}
				For the removed entries, $x=u,v$, using $\mathbf{r}^{(u)}_v=\mathbf{r}^{(v)}_u=0$,
				\begin{equation}\label{eq:E_non_zero_uv}
					E_{u,v}=E_{v,u} \;=\; \frac{1}{n-1}\left(\sqrt{\,1-1\,}-1\right) \;=\; -\,\frac{1}{n-1}=-p,
				\end{equation}
				and $E_{u,u}=E_{v,v}=0$. It corresponds to removing the edge between $u$ and $v$, whose original weight was exactly $p=\frac{1}{n-1}$. This brings a negative contribution to the eigenvalues variation described by Eq.\eqref{eq:bauer_fike_C}.
				
				For entries involving unmarked nodes different from $u$ and $v$, i.e. $x\neq u,v$, we expand the term $\sqrt{1+\mathbf{r}^{(u/v)}_x}$ using a first order Taylor approximation. 
				
				Using $\sqrt{1+\varepsilon}=1+\tfrac{\varepsilon}{2}+O(\varepsilon^2)$ we have
				
				\begin{equation}\label{eq:E_non_zero_others}
					E_{x,u} \;=\; \frac{1}{2(n-1)}\,\mathbf{r}^{(u)}_x \;+\; O\!\left(\frac{(\mathbf{r}^{(u)}_x)^2}{n}\right),
					\qquad
					E_{x,v} \;=\; \frac{1}{2(n-1)}\,\mathbf{r}^{(v)}_x \;+\; O\!\left(\frac{(\mathbf{r}^{(v)}_x)^2}{n}\right)
				\end{equation}
				Considering that $\mathbf{r}_x^{(u)}=\mathbf{r}^{(v)}_x=\frac{1}{n-2}\;\forall x\neq u,v$, we have 
				\begin{equation}
					E_{x,u}=E_{x,v}=E_{u,x}=E_{v,x}\approx \frac{1}{n^2} + O\!\left(\frac{1}{n^3}\right),
				\end{equation}
				these are positive contributions to the variation of the eigenvalues in Eq.\eqref{eq:bauer_fike_C}.\qedhere
			\end{proof}
			
			\begin{proof}[Proof of \Cref{prop:E_structure} (ii)]
				\medskip\noindent\textit{\textbf{Structure and norm of $E = \tilde{C} - C$: Spectral norm.}}
				The explicit form shows that $E$ is symmetric, supported only on $\{u,v\}$ rows/columns, and of rank at most $2$. In particular,
				\begin{equation}
					\|E\|_2 \;=\; \Theta\!\left(\frac{1}{n}\right),
				\end{equation}
				with the leading contribution coming from $E_{u,v}=-\frac{1}{n-1}=-p$, while the other entries from the redistribution are smaller by a factor $\sim 1/n-2$. Thus we can write from Eq.\eqref{eq:bauer_fike_C}:
				\begin{equation}
					\max_i |\lambda_i(\tilde{C})-\lambda_i(C)| = O\!\left(\frac{1}{n}\right).
				\end{equation} \qedhere
			\end{proof}
		}
	%SECOND APPENDIX END-------------------------------------------------------------------------------------
	
	Consequently, the eigenvalues of $\tilde{C}$ and $C$ satisfy
	\begin{equation}
		\max_i\,|\lambda_i(\tilde{C}) - \lambda_i(C)|
		= O\!\left(\frac{1}{n}\right)
	\end{equation}
	in a first order approximation of $\|E\|_2$. This means that, as the graph grows, the eigenvalues of the discriminant matrix, which govern the dynamics of the quantum walk, shift by an amount that vanishes as $O\!\left(\frac{1}{n}\right)$. 
	
	\paragraph{First order approximation of the eigenvalue shift.}
	
	Another way to estimate the shift of $C$'s eigenvalues is through the general matrix perturbation theory~\cite{Stewart1990} and the definition of the Rayleigh quotient~\cite{HornJohnson2012}. Let
	$C \in \mathbb{R}^{n\times n}$ be a real symmetric matrix and $E$ a symmetric perturbation, with $\tilde{C} = C + E$. Denote by $\lambda_i$ a simple eigenvalue of $C$, with corresponding normalized eigenvector $\mathbf{v}_i$, and by $\tilde{\lambda}_i = \lambda_i(\tilde{C})$ the perturbed eigenvalue. For any unit vector $\mathbf{x}$, the Rayleigh quotient of $\mathbf{x}$ with respect to $C$ is
	\begin{equation}
		R_C(\mathbf{x}) = \mathbf{x}^{\!T} C\,\mathbf{x},
	\end{equation}
	and when $\mathbf{x} = \mathbf{v}_i$ is an eigenvector it reproduces the corresponding eigenvalue, $R_C(\mathbf{v}_i) = \lambda_i$. The classical first-order expansion of the perturbed eigenvalue follows from the general matrix perturbation theorem (Theorem~2.3, p.~64 in~\cite{Stewart1990}), giving the following result.
	
	\begin{lemma}[First-order eigenvalue shift]
		\label{lem:rayleigh_shift}
		Let $\lambda_i$ be a simple eigenvalue of $C$ with normalized eigenvector $\mathbf{v}_i$, and let $\tilde{C} = C + E$ with $E$ symmetric. Then
		\begin{equation}
			\Delta\lambda_i
			\;=\; \mathbf{v}_i^{\!T} E\,\mathbf{v}_i + O(\|E\|_2^2).
		\end{equation}
		Hence, the leading-order variation of a simple eigenvalue under a symmetric perturbation $E$ is precisely the Rayleigh quotient of $E$ evaluated on the unperturbed eigenvector.
	\end{lemma}
	
	In our setting, the relevant eigenvalue is the simple eigenvalue $\lambda_\star = \frac{n-m-1}{n-1}$, whose normalized eigenvector for the complete graph is
	\begin{equation}
		\mathbf{v}_\star = \frac{1}{\sqrt{n-m}}\,\mathbf{1}_U,
	\end{equation}
	with $U = V \setminus M$ the set of unmarked nodes. Applying \Cref{lem:rayleigh_shift}, the first-order shift of $\lambda_\star$ under the perturbation $E = \tilde{C} - C$ is
	\begin{align}\notag
		\Delta\lambda_\star& = \; \mathbf{v}_\star^{\!T} E\,\mathbf{v}_\star\;+\;O(\|E\|_2^2)\\
		&=\; \frac{1}{n-m}\,\mathbf{1}_U^{\!T} E\,\mathbf{1}_U \;+\; O(\|E\|_2^2).
	\end{align}
	Since $\|E\|_2 = O(1/n)$, the quadratic correction $O(\|E\|_2^2) = O(1/n^2)$ is negligible, and the sign of the first-order term $\mathbf{v}_\star^{\!T} E\,\mathbf{v}_\star$ determines the direction
	of the spectral shift. A negative value implies a systematic decrease of the gap eigenvalue under the perturbation; a positive value implies an increase. Since we are interested in characterizing this change both qualitatively and quantitatively, we need to evaluate this sign, which is the content of \Cref{prop:negative_shift}.

	\paragraph{Sign and exact eigenvalue shift.}
	We now establish that the first-order perturbation of the gap eigenvalue is strictly negative, which is the key step in proving that the removal of a single edge from $K_n$ always reduces the relevant spectral quantity controlling the quantum walk dynamics. The argument relies on bounding the Rayleigh quotient of $E$ with respect to the unperturbed gap eigenvector $\mathbf{v}_\star$; details and proof of the following proposition, together with the exact magnitude of the shift, are given in \nameshortref{app:gap_eigenvalue_shift}.
	\begin{proposition}[Negative shift of the gap eigenvalue]
		\label{prop:negative_shift}
		Let $E = \tilde{C} - C$ be the perturbation of the discriminant matrix induced by the removal of the undirected edge $\{u,v\}$ from $K_n$, with $p = \frac{1}{n-1}$, under any redistribution scheme satisfying $\sum_{x\in V\setminus\{u,v\}}\mathbf{r}^{(u)}_x = \sum_{x\in V\setminus\{u,v\}}\mathbf{r}^{(v)}_x = 1$. Let $\mathbf{v}_\star =
		\frac{1}{\sqrt{n-m}}\mathbf{1}_U$ be the eigenvector associated with the gap eigenvalue $\lambda_\star$ of $C$. Then the Rayleigh quotient of $E$ with respect to $\mathbf{v}_\star$ satisfies
		\begin{equation}
			\mathbf{v}_\star^{\!T} E\, \mathbf{v}_\star
			\;<\; 0
		\end{equation}
		strictly. Consequently, the first-order (linear) approximation to the shift of the gap eigenvalue,
		\begin{equation}
			\tilde\lambda_\star - \lambda_\star = \mathbf{v}_\star^{\!T} E\,\mathbf{v}_\star + O\!\left(\|E\|_2^2\right),
		\end{equation}
		is strictly negative at leading order, confirming at this order the conjecture stated in~\cite{GIORDANO2026170305}. The exact magnitude of the shift, together with a proof of its strict negativity to all orders under the uniform redistribution assumption, is established in \Cref{prop:exact_shift}.
	\end{proposition}

	This evaluation is essential for rigorously establishing the findings presented in \cite{GIORDANO2026170305}. In that work, the decrease of the simple eigenvalue $\lambda_\star$ was extrapolated from numerical simulations and theoretically assessed from an analogous case involving the adjacency matrix. Here, we provide a formal proof for our previous conjecture: specifically, that the eigenvalue $\tilde{\lambda}_\star$ cannot exceed $\lambda_\star$. A reformulation of the completeness testing algorithm of~\cite{GIORDANO2026170305} on this analytical basis is given in \nameshortref{app:algorithm_reformulation}.
	
	Beyond its sign, we compute $\Delta\lambda_\star$ exactly under the uniform redistribution assumption, and then extract its asymptotic behavior for large $n$. Consistently with the standard Szegedy quantum walk construction, marking a node only replaces its own outgoing column with a self-loop and leaves the transition weights of unmarked nodes unchanged; we therefore assume that the redistribution vectors distribute the weight $p = \frac{1}{n-1}$ of the removed edge uniformly over \emph{all} remaining nodes other than $u$ and $v$, marked and unmarked alike:
	\begin{equation}
		\mathbf{r}^{(u)}_x = \mathbf{r}^{(v)}_x = \frac{1}{n-2}, \qquad x \in V \setminus \{u,v\}.
	\end{equation}
	
	\begin{proposition}[Exact shift of the gap eigenvalue]\label{prop:exact_shift}
		Under the uniform redistribution assumption, with the weight $p$ spread uniformly over all $n-2$ remaining nodes $V\setminus\{u,v\}$ (marked and unmarked alike), the first-order shift of the gap eigenvalue is
		\begin{equation}\label{eq:exact_rayleigh_shift}
			\Delta\lambda_\star
			\;=\; -\,\frac{2p}{k\,(n-2)}\left[\,m+(k-2)\left(\sqrt{n-1}-\sqrt{n-2}\right)^{2}\right],
		\end{equation}
		where $k := n-m$ is the number of unmarked nodes. The bracketed quantity is a sum of two nonnegative terms, strictly positive for every $n\ge3$ and every $0\le m\le n-2$; in particular, $\Delta\lambda_\star<0$ for all finite $n$, confirming the conjecture stated in~\cite{GIORDANO2026170305} to all orders.
	\end{proposition}

	\begin{proposition}[Asymptotic expansion of $\Delta\lambda_\star$]
		\label{prop:asymptotic_shift}
		Setting $k=n-m$, expression~\eqref{eq:exact_rayleigh_shift} decomposes exactly as $|\Delta\lambda_\star| = T_A + T_B$, with
		\begin{equation}\label{eq:delta_lambda_star_explicit}
			T_A = \frac{2m}{(n-1)(n-2)\,k}
			\qquad\text{(exact)},
			\qquad
			T_B = \Theta\!\left(\frac{1}{n^3}\right)
		\end{equation}
		uniformly over all $k\ge3$. Consequently, the asymptotic order of $\Delta\lambda_\star$ depends on the regime of $m$:
		\begin{itemize}
			\item \textbf{Few marked nodes} ($m=o(n)$): $|\Delta\lambda_\star| = \Theta(1/n^3)$ if $m=O(1)$, and $|\Delta\lambda_\star| = \Theta(m/n^3)$ if $m\to\infty$ with $m=o(n)$;
			\item \textbf{Marked nodes in an intermediate range} ($\varepsilon n\le m\le(1-\varepsilon)n$ for some $\varepsilon\in(0,\tfrac12]$): $|\Delta\lambda_\star| = \Theta(1/n^2)$;
			\item \textbf{Few unmarked nodes} ($k=n-m=o(n)$): 
			\begin{equation}\label{eq:delta_lambda_star_explicit_k}
				|\Delta\lambda_\star| = \Theta\!\left(\frac{1}{n\,k}\right),
			\end{equation}
			in particular $|\Delta\lambda_\star|=\Theta(1/n)$ when $k=\Theta(1)$.
		\end{itemize}
	\end{proposition}
	
	%THIRD APPENDIX BEGIN----------------------------------------------------------------------------------------------------------------------------------------
	\toappendix{
		\stepcounter{AppCount}
		\subsection*{\Alph{AppCount}. \space Gap eigenvalue shift: sign of the shift, exact evaluation and first order approximation of $\Delta \lambda_\star$}\labelshort[\Alph{AppCount}]{app:gap_eigenvalue_shift}
			\begin{proof}[Proof of \Cref{prop:negative_shift}]
				\medskip\noindent\textbf{Negative shift of the gap eigenvalue.}
				For undirected removal of the edge $\{u,v\}$ with a column-stochastic	matrix, $E$ has support only on rows/columns $u$ and $v$; the nonzero entries are given in Eq.\eqref{eq:E_non_zero_general}. Since $E$ is symmetric and vanishes outside rows/columns $\{u,v\}$, the double sum $\sum_{i,j\in U}E_{ij}$ splits into the $\{u,v\}\times\{u,v\}$ block and the two cross blocks $\{u,v\}\times(U\setminus\{u,v\})$ and $(U\setminus\{u,v\})\times\{u,v\}$, which coincide by symmetry:
				\begin{equation}
					\mathbf{1}_U^{\!T} E\,\mathbf{1}_U = \sum_{i,j\in U} E_{ij}= (E_{uu}+E_{uv}+E_{vu}+E_{vv})+ 2\!\!\sum_{x\in U\setminus\{u,v\}}\!\!\bigl(E_{ux}+E_{vx}\bigr).
				\end{equation}
				Using $E_{uu}=E_{vv}=0$ and $E_{uv}=E_{vu}=-p$ (Eq.\eqref{eq:E_non_zero_uv}), the first parenthesis equals $-2p$. Writing $\sum_{j\in U}E_{uj} = E_{u,v}+\sum_{x\in U\setminus\{u,v\}}E_{ux} = -p+\sum_x E_{ux}$ (and analogously for $v$), we substitute $\sum_{x\in U\setminus\{u,v\}}E_{ux}$ and $\sum_{x\in U\setminus\{u,v\}}E_{vx}$ respectively with $\sum_{j\in U}E_{uj} +p$ and $\sum_{j\in U}E_{vj} +p$ and obtain :
				\begin{equation}
					\mathbf{1}_U^{\!T} E\,\mathbf{1}_U= -2p + 2\Bigl[\Bigl(\sum_{j\in U}E_{uj}+p\Bigr)+\Bigl(\sum_{j\in U}E_{vj}+p\Bigr)\Bigr]= 2\sum_{j\in U}E_{uj} + 2\sum_{j\in U}E_{vj} + 2p. \label{eq:corrected_identity_negshift}
				\end{equation}
				Each row sum is, from Eqs.\eqref{eq:E_non_zero_uv},\eqref{eq:E_non_zero_general}:
				\begin{equation}
					\sum_{j\in U}E_{uj}	= -p + p\sum_{x\in U\setminus\{u,v\}} \!\!\left(\sqrt{1+\mathbf{r}^{(u)}_x}-1\right),
				\end{equation}
				and analogously for the row of $v$. Recall that $\mathbf{r}^{(u)}$ redistributes the removed weight $p$ over \emph{all} remaining nodes $V\setminus\{u,v\}$, marked and unmarked alike (consistently with the standard Szegedy marking procedure, under which unmarked nodes retain unrestricted outgoing transitions to every other node), so that $\sum_{x\in V\setminus\{u,v\}}\mathbf{r}^{(u)}_x=1$. Restricting the sum to the unmarked subset $U\setminus\{u,v\}\subseteq V\setminus\{u,v\}$ therefore gives only
				\begin{equation}\label{eq:partial_mass_bound}
					\sum_{x\in U\setminus\{u,v\}}\mathbf{r}^{(u)}_x \;\le\; \sum_{x\in V\setminus\{u,v\}}\mathbf{r}^{(u)}_x = 1,
				\end{equation}
				with equality only if no redistributed weight is sent to marked nodes (e.g.\ when $m=0$). Using the elementary inequality $\sqrt{1+\varepsilon}-1 <	\tfrac{\varepsilon}{2}$, strict for all $\varepsilon > 0$ by strict concavity of $\sqrt{\,\cdot\,}$, together with~\eqref{eq:partial_mass_bound}:
				\begin{equation}
					p\sum_{x\in U\setminus\{u,v\}}
					\left(\sqrt{1+\mathbf{r}^{(u)}_x}-1\right)
					\;<\;
					\frac{p}{2}\sum_{x\in U\setminus\{u,v\}} \mathbf{r}^{(u)}_x
					\;\le\; \frac{p}{2},
				\end{equation}
				so $\sum_{j\in U}E_{uj} < -p/2$, and similarly $\sum_{j\in U}E_{vj} < -p/2$. Substituting into Eq.\eqref{eq:corrected_identity_negshift}:
				\begin{equation}
					\mathbf{1}_U^{\!T} E\,\mathbf{1}_U
					\;<\; 2\left(-\frac p2\right)+2\left(-\frac p2\right)+2p \;=\; 0,
				\end{equation}
				so that
				\begin{equation}
					\mathbf{v}_\star^{\!T} E\,\mathbf{v}_\star
					= \frac{1}{n-m}\,\mathbf{1}_U^{\!T} E\,\mathbf{1}_U
					\;<\; 0
				\end{equation}
				strictly, for any redistribution scheme satisfying $\sum_{x\in V\setminus\{u,v\}}\mathbf{r}^{(u)}_x=1$, irrespective of how the weight is split between marked and unmarked targets. The Rayleigh quotient exact magnitude under uniform redistribution is computed in \Cref{prop:exact_shift}.
				
				\medskip\noindent\textbf{\textit{From Rayleigh quotient to eigenvalue shift.}}
				Since $C$ and $\tilde{C} = C + E$ are both real symmetric, the	first-order perturbation theory for simple eigenvalues (see	e.g.~\cite{Kato1966, Stewart1990}) gives:
				\begin{equation}
					\tilde\lambda_\star
					= \lambda_\star
					+ \mathbf{v}_\star^{\!T} E\,\mathbf{v}_\star
					+ O\!\left(\|E\|_2^2\right),
				\end{equation}
				where $\mathbf{v}_\star$ is the unit eigenvector of $C$ for	$\lambda_\star$. Since $\mathbf{v}_\star^{\!T} E\,\mathbf{v}_\star<0$ strictly, the first-order (linear) approximation to the eigenvalue shift is strictly negative, confirming at this order the conjecture stated in~\cite{GIORDANO2026170305}, where the decrease of $\lambda_\star$ was extrapolated from numerical simulations. Establishing the strict negativity of the full shift $\tilde\lambda_\star-\lambda_\star$, together with its exact rate in $n$ and $m$, is carried out directly (without invoking the generic $O(\|E\|_2^2)$ remainder bound) in \Cref{prop:exact_shift,prop:asymptotic_shift} below.\qedhere
			\end{proof}

			\begin{proof}[Proof of \Cref{prop:exact_shift}]
				\medskip\noindent\textbf{Exact shift of the gap eigenvalue.}
				Under the uniform redistribution assumption, the weight $p$ is spread uniformly over all $n-2$ remaining nodes $V\setminus\{u,v\}$ (marked and unmarked alike), so that $\mathbf{r}^{(u)}_x=\mathbf{r}^{(v)}_x=\frac1{n-2}$ for every $x\in V\setminus\{u,v\}$. As shown in the proof of \Cref{prop:E_structure}, the prefactor $\sqrt{P_{u,x}}$ in $E_{x,u}=\sqrt{P_{u,x}}\bigl(\sqrt{P_{x,u}+\Delta_{x,u}}-\sqrt{P_{x,u}}\bigr)$ vanishes identically whenever $x$ is marked, so only the $n-m-2=k-2$ unmarked targets $x\in U\setminus\{u,v\}$ contribute a nonzero entry, each equal to
				\begin{equation}
					E_{x,u}=E_{x,v}=q, \qquad q:=p\!\left(\sqrt{1+\tfrac1{n-2}}-1\right) = \frac{p\,d}{\sqrt{n-2}}, \qquad d:=\sqrt{n-1}-\sqrt{n-2},
				\end{equation}
				(rationalizing $\sqrt{1+\tfrac1{n-2}}=\sqrt{\tfrac{n-1}{n-2}}$), with all remaining entries zero except $E_{u,v}=E_{v,u}=-p$. Since $\mathbf{v}_\star=\frac1{\sqrt k}\mathbf{1}_U$, the first-order shift is
				\begin{equation}
					\Delta\lambda_\star = \frac1k\,\mathbf{1}_U^{\!T}E\,\mathbf{1}_U = \frac1k\sum_{i,j\in U}E_{ij}.
				\end{equation}
				As in the proof of \Cref{prop:negative_shift}, the double sum splits into the $\{u,v\}\times\{u,v\}$ block and the two symmetric cross blocks, each containing $k-2$ nonzero terms equal to $q$:
				\begin{equation}
					\mathbf{1}_U^{\!T}E\,\mathbf{1}_U = -2p + 4(k-2)\,q = -2p + \frac{4(k-2)\,p\,d}{\sqrt{n-2}},
				\end{equation}
				and hence
				\begin{equation}\label{eq:exact_rayleigh_shift_proof}
					\Delta\lambda_\star = \frac{2p}{k}\left[\frac{2(k-2)\,d}{\sqrt{n-2}}-1\right].
				\end{equation}
				
				\medskip\noindent\textbf{\textit{Simplification.}}
				Multiplying $d$ by its conjugate, $d\left(\sqrt{n-1}+\sqrt{n-2}\right)=(n-1)-(n-2)=1$; writing the conjugate factor as $\sqrt{n-1}+\sqrt{n-2}=d+2\sqrt{n-2}$ and substituting gives $d(d+2\sqrt{n-2})=1$, i.e.
				\begin{equation}\label{eq:d_identity}
					\frac{2d}{\sqrt{n-2}} = \frac{1-d^2}{n-2}.
				\end{equation}
				Substituting~\eqref{eq:d_identity} into the bracket of Eq.\eqref{eq:exact_rayleigh_shift_proof}:
				\begin{align}
					\frac{2(k-2)d}{\sqrt{n-2}}-1
					&= \frac{(k-2)(1-d^2)}{n-2}-1
					= \frac{(k-2)-(n-2)-(k-2)d^2}{n-2} \notag\\
					&= \frac{(k-n)-(k-2)d^2}{n-2}
					= -\frac{m+(k-2)d^2}{n-2},
				\end{align}
				using $k-n=-m$ in the last step. Substituting back into Eq.\eqref{eq:exact_rayleigh_shift_proof} yields the exact closed form, valid for every $n\ge3$ and every $0\le m\le n-2$:
				\begin{equation}\label{eq:exact_rayleigh_shift_closed}
					\Delta\lambda_\star \;=\; -\,\frac{2p}{k\,(n-2)}\Bigl[\,m + (k-2)\left(\sqrt{n-1}-\sqrt{n-2}\right)^{2}\Bigr].
				\end{equation}
				
				\medskip\noindent\textbf{\textit{Sign.}}
				The squared bracket in Eq.\eqref{eq:exact_rayleigh_shift_closed} is a sum of two non-negative terms: $m\ge0$, and $(k-2)\ge0$ with $(\sqrt{n-1}-\sqrt{n-2})^2>0$ strictly (since $n-1\neq n-2$). If $k>2$, the second term is strictly positive; if $k=2$ it means that the only unmarked nodes in the graph are $u$ and $v$, then clearly $m=n-2>0$ for $n\ge3$, so the first term is strictly positive and greater than the second one. In either case the bracket is strictly positive, so
				\begin{equation}
					\Delta\lambda_\star < 0 \qquad \text{for every finite } n\ge3 \text{ and } 0\le m\le n-2,
				\end{equation}
				confirming the claim of the proposition, uniformly over all regimes of $m$.\qedhere
			\end{proof}

	\begin{proof}[Proof of \Cref{prop:asymptotic_shift}]
		\medskip\noindent\textbf{Asymptotic expansion of $\Delta\lambda_\star$.}
		From the exact closed form~\eqref{eq:exact_rayleigh_shift_closed}, write $\Delta\lambda_\star=-(T_A+T_B)$ with
		\begin{equation}
			T_A := \frac{2pm}{k(n-2)} = \frac{2m}{(n-1)k(n-2)},
			\qquad
			T_B := \frac{2p(k-2)}{k(n-2)}\,d^2 = \frac{2(k-2)}{k(n-1)(n-2)}\,d^2,
		\end{equation}
		where $d:=\sqrt{n-1}-\sqrt{n-2}$; both $T_A,T_B\ge0$, so $|\Delta\lambda_\star|=T_A+T_B$.
		
		Regarding the term $T_B$, we recall from the proof of \Cref{prop:exact_shift} that $d\left(\sqrt{n-1}+\sqrt{n-2}\right)=1$. Since $n-2<n-1\le n$ for $n\ge3$, we have $2\sqrt{n-2}\le\sqrt{n-1}+\sqrt{n-2}\le2\sqrt n$, and therefore
		\begin{equation}\label{eq:d2_sandwich}
			\frac1{4n} \;\le\; d^2 \;\le\; \frac1{4(n-2)}, \qquad n\ge3.
		\end{equation}
		In the same way, $(k-2)/k=1-2/k$ increases while $k$ increases, and for the minimum $k=3$ it has value $1/3$, thus
		\begin{equation}\label{eq:km2_sandwich}
			\frac13 \;\le\; \frac{k-2}{k} \;<\; 1, \qquad k\ge3.
		\end{equation}
		Substituting~\eqref{eq:d2_sandwich} and~\eqref{eq:km2_sandwich},  into $T_B=\dfrac{2(k-2)}{k(n-1)(n-2)}\,d^2$ gives, for every $n\ge3$ and $k\ge3$,
		\begin{equation}\label{eq:TB_sandwich}
			\frac1{6\,n(n-1)(n-2)} \;\le\; T_B \;\le\; \frac1{2(n-1)(n-2)^2},
		\end{equation}
		so that $T_B=\Theta(1/n^3)$, uniformly for every $k\ge3$, i.e.\ for every $0\le m\le n-3$.
		
		Combining, for every $n\ge3$ and $0\le m\le n-3$ (equivalently $k\ge3$):
		\begin{equation}\label{eq:master_asymptotic}
			|\Delta\lambda_\star| \;=\; \frac{2m}{(n-1)(n-2)\,k} \;+\; \Theta\!\left(\frac1{n^3}\right).
		\end{equation}
		The boundary case $k=2$, i.e.\ $m=n-2$, is excluded from~\eqref{eq:TB_sandwich} since then $T_B=0$ exactly and $|\Delta\lambda_\star|=T_A=\frac{2(n-2)}{(n-1)(n-2)\cdot2}=\frac1{n-1}$.
		
		\medskip\noindent\textbf{\textit{Specialization to the three regimes.}}
		We specify~\eqref{eq:master_asymptotic} in the regimes of $m$ used throughout this work.
		\begin{itemize}
			\item \textbf{\textit{Few marked nodes}} ($m=o(n)$): then $k=\Theta(n)$, so $T_A=\Theta(m/n^3)$. If $m=\Theta(1)$ (in particular $m=0$), $T_A=O(1/n^3)$ is of the same order as $T_B$, and
			\begin{equation}
				|\Delta\lambda_\star| = \Theta\!\left(\frac1{n^3}\right).
			\end{equation}
			If instead $m\to\infty$ with $m=o(n)$, $T_A$ dominates and
			\begin{equation}
				|\Delta\lambda_\star| = \Theta\!\left(\frac{m}{n^3}\right).
			\end{equation}
			\item \textbf{\textit{Marked nodes in an intermediate range}} ($\varepsilon n\le m\le(1-\varepsilon)n$, hence also $\varepsilon n\le k\le(1-\varepsilon)n$): here $T_A=\Theta(1/n^2)$ (both $m$ and $k$ are $\Theta(n)$), which strictly dominates $T_B=\Theta(1/n^3)$, giving
			\begin{equation}
				|\Delta\lambda_\star| = \Theta\!\left(\frac1{n^2}\right),
			\end{equation}
			with constants depending on $\varepsilon$.
			\item \textbf{\textit{Few unmarked nodes}} ($k=n-m=o(n)$, so $m=n-k=\Theta(n)$): then $T_A=\Theta(1/(nk))$, which dominates $T_B=\Theta(1/n^3)$ since $k=o(n)\Rightarrow1/(nk)\gg1/n^2>1/n^3$, giving
			\begin{equation}
				|\Delta\lambda_\star| = \Theta\!\left(\frac1{nk}\right).
			\end{equation}
			In particular $|\Delta\lambda_\star|=\Theta(1/n)$ when $k=\Theta(1)$.
		\end{itemize}
		This establishes~\eqref{eq:master_asymptotic} together with its specializations, completing the proof.\qedhere
	\end{proof}
	}
	%THIRD APPENDIX END----------------------------------------------------------------------------------------------------------------------------------------	
	The proofs of \Cref{prop:exact_shift} and \Cref{prop:asymptotic_shift} are given in \nameshortref{app:gap_eigenvalue_shift}. As established there, the exact shift decomposes into two nonnegative contributions, $|\Delta\lambda_\star|=T_A+T_B$, with $T_A=\frac{2m}{(n-1)(n-2)k}$ known exactly and $T_B=\Theta(1/n^3)$ uniformly in $m$ (Eq.\eqref{eq:delta_lambda_star_explicit}). 
	
	\paragraph{Physical interpretation.} The strict negativity of $\Delta\lambda_\star$ has a transparent reading: removing a link can only loosen the spectral gap of the walk, never tighten it. The gap eigenvector $\mathbf v_\star=\mathbf 1_U/\sqrt{n-m}$ is uniformly spread over the unmarked nodes, so it feels the missing edge as a local depletion of connectivity, and the walk responds by lowering $\lambda_\star$. The reason the effect is so weak is equally physical. The perturbation $E$ lives entirely on the two columns $u$ and $v$, and its coupling to the marked block vanishes identically (the prefactor $\sqrt{P_{u,x}}$ is zero whenever $x$ is marked, see \Cref{prop:E_structure}); the defect therefore never talks to the marked subspace, and only the $k-2$ unmarked neighbors carry a nonzero entry. This decoupling is what suppresses the shift from the naive $\Theta(1/n)$ scale of $\|E\|_2$ down to $\Theta(1/n^2)$ in the search regime $m=\Theta(n)$: a single missing link is the minimal structural defect precisely because its spectral footprint is one order smaller than the perturbation that produces it. The two terms $T_A\propto m$ and $T_B=\Theta(1/n^3)$ separate the two ways the gap can react: $T_A$ counts how much of the eigenvector weight sits on genuinely marked-competing nodes, while $T_B$ is the residual geometric response of the pure redistribution. It is their interplay that sets the crossover between the $1/n^2$ (search) and $1/n^3$ (few marked nodes) regimes.
	
	Having established the scaling in the three general regimes above, we now specialize to the case relevant for the search algorithm. Introducing the linear optimality condition $m = \tfrac{n-1}{a}$
	from~\cite{GIORDANO2026170305} (which falls within the intermediate regime just discussed, since $m/n\to1/a\in(0,1)$), so that $k=n-m=\tfrac{(a-1)n+1}{a}$, the dominant contribution $T_A$ becomes
	\begin{align}\notag
		T_A
		&\;=\;
		\frac{2}{(n-2)\,\left((a-1)n+1\right)}\\
		&\;=\;
		\frac{2}{a-1}\,\frac{1}{n^2}
		\;+\;O\!\left(\frac{1}{n^3}\right),
	\end{align}
	and since $T_B=\Theta(1/n^3)$ is subdominant here, $\Delta\lambda_\star=-T_A+O(1/n^3)$. In particular, for $a=1.44512$,
	\begin{equation}
		\Delta\lambda_\star \;=\; -\,\frac{4.493}{n^2} + O\!\left(\frac{1}{n^3}\right).
	\end{equation}

	Thus, the first-order eigenvalue shift remains strictly negative throughout, with a magnitude that grows from $\Theta(1/n^3)$ in the few-marked-nodes regime up to $\Theta(1/n)$ when only a bounded number of nodes remain unmarked, and equals $\Theta(1/n^2)$ in the intermediate regime $m=\Theta(n)$ relevant for the search algorithm, confirming throughout that, under undirected edge removal with uniform redistribution, the gap eigenvalue decreases.

	\subsection{From eigenvalue shifts to phase shifts}\label{subsec:from_eigenvalue_shift_to_phase_shift}
	Let us recall from \cite{Portugal_book} that the eigenvalues of the quantum walk operator in Eq.\eqref{eq:quantum_walk_operator} for the complete graph case are analytically calculated and that only specific eigenspaces are involved in the dynamics given by $W$. The set $\sigma(W)$ of eigenvalues and corresponding eigenvectors is:
	
	\begin{align}
		\sigma(W) = \Big\{ & e^{\pm 2i \theta_M}=1, \quad \to \{ \ket{\omega_{M,j}^\pm}\}_{j \in \left[n-m+1, n\right]}; \nonumber \\
		& e^{\pm 2i\theta_{\star}}, \quad \to \{\ket{\omega^{+}_\star}, \ket{\omega^{-}_\star}\}; \nonumber \\
		& e^{\pm 2i\theta_B}, \quad \to \{\ket{\omega^{+}_{B,j}}, \ket{\omega^{-}_{B,j}}\}_{j \in \left[1, n-m-1\right]} \Big\}
	\end{align}
	The eigenphases associated with the eigenvalues, denoted as $\left\{\theta_M, \theta_\star, \theta_B\right\}$, are defined as follows:
	\begin{equation}\label{eq:eigenphases_complete_case}
		\cos(\theta_M) = 1 \quad \cos(\theta_\star) = \frac{n-m-1}{n-1}, \quad \cos(\theta_B) = \frac{1}{n-1}.
	\end{equation}
	
	Among these, only the eigenvectors associated with the eigenvalues $e^{\pm 2i\theta_\star}$ and $1$ are involved in the dynamics.
		
	We want to calculate the shift caused by the perturbation to the eigenvalues of the operator $W$, which means evaluating the shift on the phase $\theta_\star$ corresponding to $\lambda_\star$. Suppose the simple eigenvalue $\lambda_\star$ being shifted as $\tilde{\lambda_\star}=\lambda_\star+\Delta\lambda_\star$, we want to calculate the phase shift $\Delta\theta_\star=\tilde{\theta}_\star-\theta_\star$ with $\theta_\star=\arccos\lambda_\star$ and $\tilde{\theta_\star}=\arccos\tilde{\lambda_\star}$. The details of the	derivation are given in \nameshortref{app:phase_shift}.
	
	\begin{proposition}[First-order phase shift]
		\label{prop:phase_shift}
		Let $\theta_\star = \arccos\lambda_\star$ and $\tilde{\theta}_\star = \arccos\tilde{\lambda}_\star$. The exact relation between the phase shift and the eigenvalue shift is 
		\begin{equation}\label{eq:mvt-phase}
			\Delta\theta_\star
			= \tilde{\theta}_\star - \theta_\star
			= -\,\frac{\Delta\lambda_\star}{\sqrt{1-\xi^2}},
		\end{equation}
		for some $\xi$ between $\lambda_\star$ and $\tilde{\lambda}_\star$.	To first order in $\Delta\lambda_\star$,
		\begin{equation}\label{eq:linear_phase}
			\Delta\theta_\star
			= -\,\frac{\Delta\lambda_\star}{\sqrt{1-\lambda_\star^2}}
			\;+\; O\!\left((\Delta\lambda_\star)^2\right).
		\end{equation}
		Since $\Delta\lambda_\star < 0$, the phase shift is strictly positive: the decrease of $\lambda_\star$ causes an increase of $\theta_\star$.
	\end{proposition}

	\begin{proposition}[Asymptotic phase shift]
		\label{prop:phase_shift_explicit}
		For $\lambda_\star = \frac{n-m-1}{n-1}$, $k = n-m$, and	$\Delta\lambda_\star$ as in \Cref{prop:exact_shift}, the first-order phase shift satisfies
		\begin{equation}\label{eq:delta_theta_explicit_bound}
			\Delta\theta_\star = S_A + S_B + O\!\left(\frac1{n^2}\right),
		\end{equation}
		where
		\begin{align}\notag
			S_A &:= \frac{2}{k(n-2)}\sqrt{\frac{m}{2n-m-2}},
			\\ \label{eq:SA_SB_def}
			S_B &:= \frac{2(k-2)\left(\sqrt{n-1}-\sqrt{n-2}\right)^{2}}{k(n-2)\sqrt{m(2n-m-2)}}.
		\end{align}
		Consequently, for $m\ge1$, the asymptotic order of $\Delta\theta_\star$ depends on the regime of $m$:
		\begin{itemize}
			\item \textbf{Few marked nodes} ($m=o(n)$): $\Delta\theta_\star=\Theta(n^{-5/2})$ if $m=O(1)$, and $\Delta\theta_\star=\Theta(\sqrt m\,n^{-5/2})$ if $m\to\infty$ with $m=o(n)$;
			\item \textbf{Marked nodes in an intermediate range} ($\varepsilon n\le m\le(1-\varepsilon)n$ for some $\varepsilon\in(0,\tfrac12]$): $\Delta\theta_\star=\Theta(1/n^2)$;
			\item \textbf{Few unmarked nodes} ($k=n-m=o(n)$):
			\begin{equation}\label{eq:delta_theta_star_explicit_k}
				\Delta\theta_\star = \Theta\!\left(\frac1{nk}\right),
			\end{equation}
			in particular $\Delta\theta_\star=\Theta(1/n)$ when $k=\Theta(1)$.
		\end{itemize}
	\end{proposition}	
	
	%FOURTH APPENDIX BEGIN-------------------------------------------------------------------------------------
	\toappendix{%
		\stepcounter{AppCount}
		\subsection*{\Alph{AppCount}.\space Phase shift propositions proofs}\labelshort[\Alph{AppCount}]{app:phase_shift}
		
		\begin{proof}[Proof of \Cref{prop:phase_shift}]
			\medskip\noindent\textbf{First-order phase shift.}
			Since $\arccos$ is differentiable on $(-1,1)$, by the mean-value theorem there exists $\xi$ strictly between $\lambda_\star$ and $\tilde{\lambda}_\star$ such that
			\begin{equation}
				\tilde{\theta}_\star - \theta_\star
				= -\,\frac{\tilde{\lambda}_\star - \lambda_\star}{\sqrt{1-\xi^2}},
				\qquad
				|\Delta\theta_\star|
				= \frac{|\Delta\lambda_\star|}{\sqrt{1-\xi^2}},
			\end{equation}
			which establishes~\eqref{eq:mvt-phase}. This relation is exact and shows that the phase shift is proportional to the eigenvalue shift, with a factor depending on the intermediate value $\xi$.
			When the perturbation is small, $\tilde{\lambda}_\star \approx \lambda_\star$ and hence $\xi \approx \lambda_\star$. Using $\frac{d}{d\lambda}\arccos\lambda = -1/\sqrt{1-\lambda^2}$, the first-order expansion gives
			\begin{equation}
				\Delta\theta_\star
				= -\,\frac{\Delta\lambda_\star}{\sqrt{1-\lambda_\star^2}}
				+ O\!\left((\Delta\lambda_\star)^2\right),
			\end{equation}
			which is~\eqref{eq:linear_phase}. The minus sign reflects that $\arccos$ is strictly decreasing, so a negative $\Delta\lambda_\star$ produces a positive $\Delta\theta_\star$. The denominator
			$\sqrt{1-\lambda_\star^2}$ quantifies how strongly a spectral perturbation of $C$ translates into a phase perturbation of $W$.\qedhere
		\end{proof}

	\begin{proof}[Proof of \Cref{prop:phase_shift_explicit}]
		\medskip\noindent\textbf{Asymptotic phase shift.}
		For $\lambda_\star = \frac{n-m-1}{n-1}$ one computes
		\begin{equation}\label{eq:denom_phase}
			\sqrt{1-\lambda_\star^2} = \sqrt{1 - \left(\frac{n-m-1}{n-1}\right)^{\!2}} = \frac{\sqrt{m\,(2n-m-2)}}{n-1}.
		\end{equation}
		Substituting into~\eqref{eq:linear_phase} gives
		\begin{equation}\label{eq:phase_shift_substitution}
			\Delta\theta_\star = -\frac{n-1}{\sqrt{m(2n-m-2)}}\,\Delta\lambda_\star	+ O\!\left((\Delta\lambda_\star)^2\right).
		\end{equation}
		Using the exact expression for $\Delta\lambda_\star$ from \Cref{prop:exact_shift} (Eq.\eqref{eq:exact_rayleigh_shift}), and $(n-1)p=1$, Eq.\eqref{eq:phase_shift_substitution} becomes
		\begin{equation}
			\Delta\theta_\star = \frac{2}{k(n-2)\sqrt{m(2n-m-2)}}\Big[\,m+(k-2)\left(\sqrt{n-1}-\sqrt{n-2}\right)^{2}\Big] + O\!\left(\frac1{n^2}\right),
		\end{equation}
		where the remainder follows from $|\Delta\lambda_\star|\le\|E\|_2=O(1/n)$ (\Cref{prop:E_structure}), so $(\Delta\lambda_\star)^2=O(1/n^2)$ irrespective of the regime of $m$. Splitting the bracket over the shared denominator and using $\dfrac{m}{\sqrt{m(2n-m-2)}}=\sqrt{\dfrac{m}{2n-m-2}}$ for the first term gives
		\begin{equation}\label{eq:delta_theta_explicit_bound_proof}
			\Delta\theta_\star = S_A + S_B + O\!\left(\frac1{n^2}\right),
		\end{equation}
		where
		\begin{equation}\label{eq:SA_SB_def_proof}
			S_A := \frac{2}{k(n-2)}\sqrt{\frac{m}{2n-m-2}},
			\qquad
			S_B := \frac{2(k-2)}{k(n-2)\sqrt{m(2n-m-2)}}\left(\sqrt{n-1}-\sqrt{n-2}\right)^{2}.
		\end{equation}
		This establishes~\eqref{eq:delta_theta_explicit_bound}.\qedhere
		
	\medskip\noindent\textbf{\textit{Specialization to the three regimes.}}
	We read off Eq.\eqref{eq:delta_theta_explicit_bound_proof} in the regimes of $m$ used throughout this work.
	\begin{itemize}
		\item \textbf{\textit{Few marked nodes}} ($m=o(n)$): then $k=\Theta(n)$ and $\sqrt{m(2n-m-2)}\sim\sqrt{mn}$, so, using $d^2\sim1/(4n)$ from \Cref{prop:asymptotic_shift},
		\begin{equation}\notag
			S_A \sim \frac2{n^2}\sqrt{\frac mn} = \Theta\!\left(\frac{\sqrt m}{n^{5/2}}\right),
			\qquad
			S_B \sim \frac2n\cdot\frac1{4n}\cdot\frac1{\sqrt{mn}} = \Theta\!\left(\frac1{n^{5/2}\sqrt m}\right).
		\end{equation}
		If $m=O(1)$, $S_A$ and $S_B$ are of the same order and $\Delta\theta_\star=\Theta(n^{-5/2})$. If instead $m\to\infty$ while $m=o(n)$, $S_A$ dominates $S_B$ (their ratio $S_A/S_B\sim m\to\infty$), and $\Delta\theta_\star=\Theta\!\left(\sqrt m\,n^{-5/2}\right)$.
		
		\item \textbf{\textit{Marked nodes in an intermediate range}} ($\varepsilon n\le m\le(1-\varepsilon)n$, hence also $\varepsilon n\le k\le(1-\varepsilon)n$): then $k=\Theta(n)$ and $\sqrt{m(2n-m-2)}\sim n$, so $S_A=\Theta(1/n^2)$ while $S_B=\Theta(1/n^3)$ is subdominant, giving $\Delta\theta_\star=\Theta(1/n^2)$.
		\item \textbf{\textit{Few unmarked nodes}} ($k=n-m=o(n)$, so $m=n-k=\Theta(n)$): then $\sqrt{m(2n-m-2)}=\Theta(n)$, so $S_A=\Theta(1/(kn))$, while $S_B=\Theta(1/n^3)$ remains subdominant since $k=o(n)$ implies $1/(kn)\gg1/n^3$. This establishes Eq.\eqref{eq:delta_theta_star_explicit_k}, and in particular $\Delta\theta_\star=\Theta(1/n)$ when $k=\Theta(1)$.
	\end{itemize}
	
	This establishes the specializations claimed in the proposition, completing the proof.\qedhere
	\end{proof}
	}
	%FOURTH APPENDIX END----------------------------------------------------------------------------------------------------------------------------------------

The scaling of $\Delta\theta_\star$ in the three regimes of $m$ is established in \Cref{prop:phase_shift_explicit}, with detailed derivation given in \nameshortref{app:phase_shift}. In particular, for a fixed (constant) number of marked nodes, $\Delta\theta_\star=\Theta(n^{-5/2})$, while in the intermediate regime relevant for the search algorithm, $m=\Theta(n)$, the scaling is $\Delta\theta_\star=\Theta(1/n^2)$.

We now specialize to the regime relevant for the search algorithm, where $m = \tfrac{n-1}{a}$ with $a = 1.44512$ from~\cite{GIORDANO2026170305}; this falls within the intermediate regime discussed above, since $m/n\to1/a\in(0,1)$. Hence:
\begin{equation}
	k \;=\; n-m \;=\; \frac{(a-1)n+1}{a}
	\;=\; \Theta(n),
\end{equation}
and
\begin{equation}
	\sqrt{m\,(2n-m-2)}
	\;=\;
	\frac{n-1}{a}\,\sqrt{\,2a-1\,}
	\;=\; \Theta(n).
\end{equation}

Substituting them into $S_A$, the dominant contribution (Eq.\eqref{eq:SA_SB_def_proof}), we get
\begin{align}
	\Delta\theta_\star=\frac{2a}{(a-1)\sqrt{2a-1}}\;\frac{1}{n^2}
	\;+\;O\!\left(\frac{1}{n^3}\right),
\end{align}
where $S_B$ contributes only at order $O(1/n^3)$ and is absorbed into the remainder. In particular, for $a=1.44512$
\begin{equation}
	\Delta\theta_\star=
	4.723\frac{1}{n^2}
	\;+\;O\!\left(\frac{1}{n^3}\right)
	=
	\Theta\!\left(\frac{1}{n^2}\right).
\end{equation}

We observe that the behavior of $\Delta\theta_\star$ for few marked nodes, in particular for $m$ constant, is compatible with the results obtained in \cite{GIORDANO2026170305} for the quantum phase estimation procedure (QPE). Specifically, the second part of the algorithm for completeness testing in \cite{GIORDANO2026170305} involves estimation of the phase $\theta_\star$ in order to distinguish complete and incomplete graphs with $m=1$, one single marked node. The number of qubits necessary in the first register of the QPE is estimated empirically based on the fitting of $\Delta\theta_\star$, where $\tilde{\theta}_\star$ is given by numerical simulations and $\theta_\star$ is known analytically. These fitting results are compatible with the asymptotic behavior we find for this regime of $m$, $\Delta\theta_\star=\Theta(n^{-5/2})$, and this finding further validates the empirical results obtained in \cite{GIORDANO2026170305}, as it is shown in the analytical derivation of the QPE qubit count in \nameshortref{app:algorithm_reformulation}.

\paragraph{Physical interpretation.} The phase $\theta_\star=\arccos\lambda_\star$ is the quantity the search dynamics actually rotate through, so it is the physically observable consequence of the spectral shift. Because $\arccos$ is strictly decreasing, the loosening of the gap ($\Delta\lambda_\star<0$) translates into a strictly positive phase shift $\Delta\theta_\star>0$: the missing link makes the walk precess slightly faster around the marked subspace. The magnitude, however, is amplified relative to the eigenvalue shift by the geometric factor $|\Delta\theta_\star/ \Delta \lambda_\star|\approx 1/\sqrt{1-\lambda_\star^2}=1/\sin\theta_\star$, and it is this factor that changes the exponent in the distinct regimes of $m$. For a macroscopic marked fraction ($m=\Theta(n)$, the search regime) $\sin\theta_\star=\Theta(1)$, the amplification is order one, and $\Delta\theta_\star$ inherits the $\Theta(1/n^2)$ scale of $\Delta\lambda_\star$.
	
\section{Eigenspaces rotation}
\label{sec:contribution_of_the_eigenvector_rotation}

While we have already extensively dedicated
\Cref{subsec:from_eigenvalue_shift_to_phase_shift} to the shift of the eigenvalue $\lambda_\star$ and the corresponding phase shift $\Delta\theta_\star = \Theta(1/n^2)$ (for $m=\Theta(n)$), we now evaluate the eigenvector
rotation due to the single edge removal. Our primary objective is to estimate the geometric reorientation of the effective subspace $\mathcal{U}_{\mathrm{eff}}$ in $\mathbb{R}^n$ and then, via the
explicit map between eigenvectors of $C$ and eigenvectors of $W$, determine the corresponding rotation of the relevant eigenvectors in $\mathcal{H}$.

\subsection{Rotation of the subspaces in $\mathbb{R}^n$}
We denote by $\angle(\mathcal{U}, \tilde{\mathcal{U}})$ the largest principal angle between $\mathcal{U}$ and $\tilde{\mathcal{U}}$, thus $\|\sin\Theta(\mathcal{U},\tilde{\mathcal{U}})\|_2 =\sin\angle(\mathcal{U},\tilde{\mathcal{U}})$. Let $\Theta(\mathcal{U}_{\mathrm{eff}},\tilde{\mathcal{U}}_{\mathrm{eff}})$ be the diagonal matrix of principal angles between $\mathcal{U}_{\mathrm{eff}}$ and $\tilde{\mathcal{U}}_{\mathrm{eff}}$. We define the rotation angle $\alpha_{\mathrm{eff}}$ as the largest principal angle between the two subspaces:
\begin{equation}
	\alpha_{\mathrm{eff}} := \angle(\mathcal{U}_{\mathrm{eff}},
	\tilde{\mathcal{U}}_{\mathrm{eff}}) = \max_j \alpha_j.
\end{equation}
Crucially, the spectral norm of the sine of the angle matrix corresponds exactly to the sine of this maximum angle:
\begin{equation}
	\|\sin\Theta(\mathcal{U}_{\mathrm{eff}},
	\tilde{\mathcal{U}}_{\mathrm{eff}})\|_2 = \sin\alpha_{\mathrm{eff}}.
\end{equation}
Therefore, bounding $\|\sin\Theta\|_2$ is equivalent to bounding the maximum rotation angle $\alpha_{\mathrm{eff}}$.

	The main tool that we use is the Davis--Kahan $\sin\Theta$ theorem~\cite{DavisKahan1970}, in the form given in \cite{Stewart1990}, Theorems~3.4 and~3.6, Chapter~V, which bounds the distance between original and perturbed invariant subspaces of a symmetric matrix in terms of the perturbation norm and the spectral gap. The details of the computation of the gaps and the application of the theorem are given in \nameshortref{app:rotation_angle}. Matrix $C$ admits an orthogonal spectral decomposition into three invariant subspaces:
\begin{equation}
	\mathbb{R}^n = \mathcal{U}_M \oplus \mathcal{U}_\star \oplus
	\mathcal{U}_B,
\end{equation}
where $\mathcal{U}_M$ corresponds to $\lambda=1$ with multiplicity $m$, $\mathcal{U}_\star$ is the one-dimensional subspace of the simple eigenvalue $\lambda_\star = \tfrac{n-m-1}{n-1}$, and $\mathcal{U}_B$ is the bulk subspace of dimension $n-m-1$ associated with $-\tfrac{1}{n-1}$. As established in \Cref{subsec:discriminant_with_marked_vertices_block_form}, the success probability $p_M(t)$ is entirely determined by the projection of $\ket{\psi_0}$ onto $\mathcal{V}_{\mathrm{eff}}$, the invariant subspace of $W$ in $\mathcal{H}$ corresponding to $\mathcal{U}_{\mathrm{eff}} = \mathcal{U}_M \oplus \mathcal{U}_\star$, thus we start by analyzing the rotation of $\mathcal{U}_{\mathrm{eff}}$. Using the exact preservation of $\mathcal{U}_{\mathrm{M}}$ together with the spectral gap $\delta_\star$ yields the following bound.

\begin{proposition}[Rotation of $\mathcal{U}_{\mathrm{eff}}$]
	\label{prop:subspace_rotation}
	Let $\alpha_{\mathrm{eff}} := \angle(\mathcal{U}_{\mathrm{eff}}, \tilde{\mathcal{U}}_{\mathrm{eff}})$ be the largest principal angle	between the unperturbed and perturbed effective subspaces of $C$ in $\mathbb{R}^n$. Since $E = 0_{(m)} \oplus E_U$, the marked block of $\tilde{C}$ is identical to that of $C$, so $\tilde{\mathcal{U}}_M = \mathcal{U}_M$ exactly and the rotation of $\mathcal{U}_{\mathrm{eff}}$ comes entirely from $\mathcal{U}_\star$. The Davis--Kahan theorem applied to $\mathcal{U}_\star$ gives:
	\begin{equation}\label{eq:subspace_rotation_bound}
		\sin\alpha_{\mathrm{eff}}
		= \|\sin\Theta(\mathcal{U}_\star,
		\tilde{\mathcal{U}}_\star)\|_2
		\leq \frac{\|E\|_2}{\delta_\star},
	\end{equation}
	where $\delta_\star = \min\!\left\{\frac{n-m}{n-1},\,\frac{m}{n-1}\right\}$. In particular, for $m = \Theta(n)$ the gap satisfies $\delta_\star = \Theta(1)$ and the bound reduces to $O(1/n)$.
\end{proposition}

The proof is given in \nameshortref{app:rotation_angle}. This proposition highlights that for $\mathcal{U}_M$ the block-diagonal structure of the perturbation has an exact
consequence,  $\tilde{\mathcal{U}}_M =\mathcal{U}_M$ exactly, with no rotation at all.

\paragraph{Rotation of $\mathcal{U}_B$ in $\mathbb{R}^n$.}
The bulk subspace $\mathcal{U}_B$ also rotates under the perturbation $E$. The spectral gap separating $\lambda_B = -\tfrac{1}{n-1}$ from the rest of the spectrum is:
\begin{align}\notag
	\delta_B &= \min\!\left\{|\lambda_B - \lambda_\star|,\,
	|\lambda_B - 1|\right\} \\
	&= \min\!\left\{\frac{n-m}{n-1},\, \frac{n}{n-1}\right\}
\end{align}
which for $m=\Theta(n)$ becomes $\delta_B= \Theta(1)$. Applying Davis--Kahan with $\|E\|_2 = O(1/n)$ and keeping the regime $m=\Theta(n)$ we have:
\begin{equation}
	\|\sin\Theta(\mathcal{U}_B, \tilde{\mathcal{U}}_B)\|_2
	\leq \frac{\|E\|_2}{\delta_B} = O\!\left(\frac{1}{n}\right).
\end{equation}
The rotation of $\mathcal{U}_B$ is therefore of the same order as that of $\mathcal{U}_\star$. However, there is a crucial structural difference: since both $C$ and $E$ are block-diagonal with $E = 0_{(m)} \oplus E_U$, the bulk eigenvectors $\mathbf{v}_{B,l}$ have support strictly on the unmarked block $U$ and this support is preserved after the perturbation. Concretely, $\tilde{C} = I_{(m)} \oplus (C_U + E_U)$, so the perturbed bulk eigenvectors $\tilde{\mathbf{v}}_{B,l}$ are eigenvectors of $C_U + E_U$ and therefore have zero components on the marked nodes both before and
after the perturbation.

This block structure constrains the rotation to stay entirely within the unmarked subspace $\mathbb{R}^k$, and has a direct consequence at the level of $\mathcal{H}$, as we will see in \Cref{prop:bulk_orthogonality}.

\subsection{Rotation of subspaces and vectors in $\mathcal{H}$}
In this subsection we analyze the effects of the perturbation on the evolution operator eigenspaces. First, we analyze how the perturbation changes the subspace $\mathcal{V}_\star$ through its eigenvectors.
\paragraph{Rotation of the	eigenvectors of $W$ in $\mathcal{V}_\star$.}
The eigenvectors of $W$ in $\mathcal{H}$ are related to the eigenvectors of $C$ in $\mathbb{R}^n$ by the explicit formula of~\cite{Portugal_book} (Theorem~11.4):
\begin{equation}
	\ket{\omega^\pm_\star} = \frac{A\mathbf{v}_\star -
		e^{\pm i\theta_\star}B\mathbf{v}_\star}{\sqrt{2}\sin\theta_\star},
	\qquad
	\ket{\tilde\omega^\pm_\star} = \frac{\tilde{A}\tilde{\mathbf{v}}_\star
		- e^{\pm i\tilde\theta_\star}\tilde{B}\tilde{\mathbf{v}}_\star}{
		\sqrt{2}\sin\tilde\theta_\star},
\end{equation}
where $A$ and $B$ have been defined in \Cref{sec:szegedy_quantum_walk_construction}, and $\tilde{A}$, $\tilde{B}$ are their perturbed counterparts. The key quantity controlling the change of the embedding of $\mathbb{R}^n$ into $\mathcal{H}$ is the spectral norm of the isometry perturbation, whose proof is given in \nameshortref{app:rotation_angle}.

With respect to $A$ and $B$ we formulate the following lemma:

\begin{lemma}[Spectral norm of the isometry perturbation]
	\label{lem:A_perturbation}
	Let $A = \sum_u|\psi_u\rangle\langle u|$ and $\tilde{A} = \sum_u|\tilde\psi_u\rangle\langle u|$ be the isometries associated with the complete graph $K_n$ with marked nodes and the graph $G_{u,v}$ obtained by removing the edge $\{u,v\}$ with uniform redistribution, respectively. Then:
	\begin{equation}
		\|\tilde{A} - A\|_2 = \Theta\!\left(\frac{1}{\sqrt{n}}\right),
	\end{equation}
	and the same holds for $\|\tilde{B} - B\|_2$.
\end{lemma}

Notice that this lemma makes no assumptions on the number of marked nodes $m$. Proof of \Cref{lem:A_perturbation} is given in  \nameshortref{app:rotation_angle}. The rotation of the eigenvectors $\ket{\omega^\pm_\star}$ in $\mathcal{H}$ receives two distinct contributions: the spectral rotation of $\mathbf{v}_\star$ in $\mathbb{R}^n$, transported to $\mathcal{H}$ via $A$ and $B$, which by \Cref{prop:subspace_rotation} contributes only $O(1/n)$; and the change of the operators
$A \to \tilde{A}$ and $B \to \tilde{B}$ themselves, which by \Cref{lem:A_perturbation} contributes $\Theta(1/\sqrt{n})$ and dominates. The proof of the following proposition is given in \nameshortref{app:rotation_angle}.

\begin{proposition}[Rotation of $\ket{\omega^\pm_\star}$ in $\mathcal{H}$]
	\label{prop:omega_star_rotation}
	Under the perturbation induced by the removal of edge $\{u,v\}$, and for $m = \Theta(n)$ so that $\sin\theta_\star = \Theta(1)$, the eigenvectors $\ket{\omega^\pm_\star}$ of $W$ in $\mathcal{V}_\star$ satisfy:
	\begin{equation}\label{eq:omega_star_rotation}
		\|\ket{\tilde\omega^\pm_\star} - \ket{\omega^\pm_\star}\|_2
		= O\!\left(\frac{1}{\sqrt{n}}\right),
	\end{equation}
	dominated by the change of the isometries $\|\tilde{A} - A\|_2 = \Theta(1/\sqrt{n})$ from \Cref{lem:A_perturbation}, while the contribution from the rotation of $\mathbf{v}_\star$ in $\mathbb{R}^n$
	is only $O(1/n)$.
\end{proposition}

\paragraph{Rotation of the subspace $\mathcal{V}_\star$ in $\mathcal{H}$.}
While \Cref{prop:omega_star_rotation} bounds the distance between the
individual eigenvectors $\ket{\omega^\pm_\star}$ and
$\ket{\tilde\omega^\pm_\star}$, it is also important to bound the
rotation of the two-dimensional invariant subspace
$\mathcal{V}_\star$ as a whole, since this quantity will be used in the next \Cref{sec:effects_on_succ_prob}. Since the eigenvalues
$e^{\pm 2i\theta_\star}$ of $W$ in $\mathcal{V}_\star$ are
well-separated from all other eigenvalues of $W$ (both $e^{\pm 2i\theta_B}$
and $1$), Theorem~VII.3.4 of~\cite{Bhatia1997} can be applied directly
to $\mathcal{V}_\star$. The proof is given in
\nameshortref{app:rotation_angle}.

\begin{proposition}[Rotation of $\mathcal{V}_\star$ in $\mathcal{H}$]
	\label{prop:V_star_rotation}
	Let $\delta_W = \min\bigl(|e^{\pm 2i\theta_\star} - 1|,\;|e^{\pm 2i\theta_\star} - e^{\pm 2i\theta_B}|\bigr)$, which becomes $\delta_W= \Theta(1)$ for $m = \Theta(n)$. Then:
	\begin{equation}\label{eq:V_star_rotation}
		\|\Pi_{\tilde{\mathcal{V}}_\star}
		- \Pi_{\mathcal{V}_\star}\|_2
		\leq \frac{\|\tilde{W} - W\|_2}{\delta_W}
		= O\!\left(\frac{1}{\sqrt{n}}\right),
	\end{equation}
	by Theorem~VII.3.4 of~\cite{Bhatia1997} applied to $W$ and $\tilde{W}$ as normal operators, combined with \Cref{cor:W_perturbation}.
\end{proposition}

This proposition establishes that the subspace $\tilde{\mathcal{V}}_\star$ is rotated by $O(1/\sqrt{n})$ with respect to $\mathcal{V}_\star$, in the regime $m=\Theta(n)$. The rotation outside this regime acquires a factor of $1/\delta_W$.
\paragraph{Invariance of $\mathcal{V}_M$ and bulk orthogonality	of $\mathcal{V}_B$.}
We now state two exact consequences of the block-diagonal structure of $E = 0_{(m)} \oplus E_U$ at the level of $\mathcal{H}$.

First, since $\tilde{\mathcal{U}}_M = \mathcal{U}_M$ exactly (\Cref{prop:subspace_rotation}) and the marked columns of $P$ are unchanged by the perturbation, the eigenvectors of $W$ in $\mathcal{V}_M$ are exactly preserved:
\begin{equation}
	|\tilde\omega_{M,j}\rangle = |\omega_{M,j}\rangle
	\quad \text{exactly, for all } j = 1,\ldots,m.
\end{equation}

Second, as established in Section~5.1, the rotation of $\mathcal{U}_B$ in $\mathbb{R}^n$ stays strictly within the unmarked block. At the level of $\mathcal{H}$, this has the following exact consequence.

\begin{proposition}[Bulk orthogonality]
	\label{prop:bulk_orthogonality}
	Since $E = 0_{(m)} \oplus E_U$, the perturbed matrix inherits the block-diagonal form $\tilde{C} = I_{(m)} \oplus (C_U + E_U)$. By the spectral theory of block-diagonal matrices, the perturbed bulk
	eigenvectors in $\tilde{\mathcal{V}}_B$ retain support strictly on the unmarked block and remain orthogonal to $\tilde{\mathcal{V}}_M$. Hence:
	\begin{equation}
		\Pi_M\ket{\tilde\omega_{B,l}} = 0 \quad \text{for all } l.
	\end{equation}
	Consequently, for any state $\ket{\psi} \in \tilde{\mathcal{V}}_B$, $\Pi_M\tilde{W}^t\ket{\psi} = 0$ exactly, since $\tilde{W}^t$ preserves $\tilde{\mathcal{V}}_B$.
\end{proposition}

The proof is given in \nameshortref{app:rotation_angle}. This last paragraph completes the analysis of the subspace rotation since $\mathcal{H} = \mathcal{V}_{\mathrm{eff}} \oplus \mathcal{V}_B \oplus (\mathcal{A}^\perp \cap \mathcal{B}^\perp)$, and $\tilde{\mathcal{V}}_M = \mathcal{V}_M$ exactly, the effective subspace rotation is given by the rotation of $\mathcal{V}_\star$. This will be useful in the next section to define the contributions to the probability variation under the perturbation. 
%FIFTH APPENDIX BEGIN----------------------------------------------------------------------------------------------------------------------------------------
\toappendix{%
	\stepcounter{AppCount}
	\subsection*{\Alph{AppCount}.\space Subspace and eigenvector rotation under edge removal} \labelshort[\Alph{AppCount}]{app:rotation_angle}
	In this appendix we first give an overview of two important theorems that are employed in many of the proofs in the following appendices, the Davis--Kahan bound (or theorem) and the Theorem VII.3.4 for normal operators in \cite{Bhatia1997}. Bounds of this type, controlling subspace rotation or dynamical leakage by the ratio between the perturbation norm and the spectral gap, are a recurring feature of quantum perturbation theory; see e.g.~\cite{SzaboEtAl2025} for a recent time-independent leakage bound of this form for block-diagonal effective dynamics in gapped quantum systems. Then we will proceed to the proofs for \Cref{prop:subspace_rotation}, \Cref{lem:A_perturbation}, \Cref{prop:omega_star_rotation}, \Cref{prop:V_star_rotation} and \Cref{prop:bulk_orthogonality}.
	
	\medskip
	\noindent\textbf{\textit{Davis--Kahan bound.}}
	Because both $C$ and $\tilde{C} = C + E$ are real and symmetric, we can quantify how much the eigenvectors of $C$ rotate under the perturbation $E$ using the Davis--Kahan $\sin\Theta$ theorem~\cite{DavisKahan1970}, in the form presented in~\cite{Stewart1990}, Theorems~3.4 and~3.6, Chapter~V. This theorem establishes an upper bound on the distance between original and perturbed invariant subspaces of a Hermitian matrix in terms of the norm of the perturbation and of the spectral gap.
	
	More specifically, the theorem postulates that the norm of the sine of the principal angles matrix between the two subspaces is bounded by the ratio between the perturbation norm and the minimum spectral gap of the	unperturbed matrix spectrum. This allows us to quantify the stability of the eigenvectors, ensuring that provided the eigenvalues are sufficiently separated, the estimated subspace remains close to the original one.
	
	Our application simplifies it since we are dealing with an invariant subspace of dimension one, corresponding to the simple eigenvalues $\lambda_\star$ and $\tilde{\lambda}_\star$. Let $D\in\mathbb{R}^{n\times n}$ be symmetric, and let $\tilde{D} = D + E$ with $E$ a symmetric perturbation. Denote by $\mathcal{U}$ and $\tilde{\mathcal{U}}$ the invariant subspaces of $D$ and $\tilde{D}$, respectively. If $\Theta(\mathcal{U},\tilde{\mathcal{U}})$ is the diagonal matrix of principal angles, the Davis--Kahan theorem states that
	\begin{equation}\label{eq:davis-kahan-general}
		\|\sin\Theta(\mathcal{U},\tilde{\mathcal{U}})\|_2
		\;\le\;
		\frac{\|D Z - Z F\|_2}{\delta},
	\end{equation}
	where $Z$ spans $\tilde{\mathcal{U}}$, $F$ is any Hermitian matrix of appropriate size, and $\delta$ is the minimum spectral gap separating the eigenvalues of interest from the rest of the spectrum. The matrix $R = DZ - ZF$ is called the residual; when $\tilde{D} = D + E$ with $E$ small, $R = EZ$ and the rotation is controlled directly by $\|E\|_2$. When $\mathcal{U}$ is one-dimensional, spanned by a single eigenvector $\mathbf{v}_i$ of a simple eigenvalue $\lambda_i$, the subspace angle reduces to the vector angle $\angle$ between $\mathbf{v}_i$ and $\tilde{\mathbf{v}}_i$, and~\eqref{eq:davis-kahan-general} becomes
	\begin{equation}\label{eq:davis-kahan-vector}
		\sin\angle(\mathbf{v}_i, \tilde{\mathbf{v}}_i)
		\;\le\;
		\frac{\|E\|_2}{\delta_i},
		\qquad
		\delta_i = \min_{j \neq i}|\lambda_i - \lambda_j|.
	\end{equation}
	
	\medskip
	\noindent\textbf{\textit{Bhatia's Theorem VII.3.4 for normal operators.}}
	While the Davis--Kahan theorem applies to symmetric (Hermitian)
	matrices, the walk operators $W$ and $\tilde{W}$ are unitary, hence
	normal but not Hermitian. To bound the rotation of their invariant
	subspaces, we use Theorem~VII.3.4 of~\cite{Bhatia1997}, which
	generalizes the perturbation theory of invariant subspaces to normal
	operators.
	
	Let $N_1$ and $N_2$ be normal operators on a finite-dimensional
	Hilbert space $\mathcal{H}$. Let $S_1 \subset \mathbb{C}$ and
	$S_2 \subset \mathbb{C}$ be two disjoint sets of complex numbers,
	and let $\Pi_1 = E_{N_1}(S_1)$ and $\Pi_2 = E_{N_2}(S_2)$ be the
	spectral projections of $N_1$ onto eigenvalues in $S_1$ and of $N_2$
	onto eigenvalues in $S_2$, respectively. Define the separation
	$\delta = \mathrm{dist}(S_1, S_2) = \min\{|s_1 - s_2| : s_1 \in S_1,
	\, s_2 \in S_2\}$. Then Theorem~VII.3.4 of~\cite{Bhatia1997} states
	that:
	\begin{equation}\label{eq:bhatia-VII-3-4}
		\|\Pi_1 \Pi_2\|_2 \leq \frac{1}{\delta}
		\|\Pi_1 (N_1 - N_2) \Pi_2\|_2.
	\end{equation}
	In particular, since $\|\Pi_1(N_1 - N_2)\Pi_2\|_2 \leq
	\|N_1 - N_2\|_2$:
	\begin{equation}\label{eq:bhatia-VII-3-4-simple}
		\|\Pi_1 \Pi_2\|_2 \leq \frac{\|N_1 - N_2\|_2}{\delta}.
	\end{equation}
	When $S_1$ contains the eigenvalues of $N_1$ associated with a
	subspace $\mathcal{V}$ and $S_2$ contains the eigenvalues of $N_2$
	outside the corresponding perturbed subspace $\tilde{\mathcal{V}}$,
	we have $\Pi_1 = \Pi_{\mathcal{V}}$ and
	$\Pi_2 = \Pi_{\tilde{\mathcal{V}}^\perp}$. If $\mathcal{V}$ and
	$\tilde{\mathcal{V}}$ have the same dimension, the standard identity
	for orthogonal projectors of equal rank gives:
	\begin{equation}\label{eq:projector-identity}
		\|\Pi_{\tilde{\mathcal{V}}} - \Pi_{\mathcal{V}}\|_2
		= \|\Pi_{\mathcal{V}}\Pi_{\tilde{\mathcal{V}}^\perp}\|_2,
	\end{equation}
	so that~\eqref{eq:bhatia-VII-3-4-simple} directly bounds the
	projector difference:
	\begin{equation}\label{eq:bhatia-projector-bound}
		\|\Pi_{\tilde{\mathcal{V}}} - \Pi_{\mathcal{V}}\|_2
		\leq \frac{\|N_1 - N_2\|_2}{\delta}.
	\end{equation}
	This is the key tool we use to bound the rotation of the invariant
	subspaces of $W$ in $\mathcal{H}$, where the separation $\delta$
	between eigenvalues plays the role of the spectral
	gap.

	\begin{proof}[Proof of \Cref{prop:subspace_rotation}]
		\medskip\noindent\textbf{Rotation of $\mathcal{U}_{\mathrm{eff}}$.}
		Matrix $C$ admits an orthogonal spectral decomposition into three invariant subspaces:
		\begin{equation}
			\mathbb{R}^n = \mathcal{U}_M \oplus \mathcal{U}_\star \oplus
			\mathcal{U}_B,
		\end{equation}
		where $\mathcal{U}_M$ corresponds to the marked eigenvalue $\lambda=1$ with multiplicity $m$, $\mathcal{U}_\star$ is the one-dimensional subspace of the simple eigenvalue $\lambda_\star = \tfrac{n-m-1}{n-1}$, and $\mathcal{U}_B$ is the bulk subspace of dimension $n-m-1$ associated with $-\tfrac{1}{n-1}$.
		
		\medskip
		\noindent\textbf{\textit{Exact preservation of $\mathcal{U}_M$.}}
		Since $E = 0_{(m)} \oplus E_U$, the perturbed discriminant matrix is $\tilde{C} = I_{(m)} \oplus (C_U + E_U)$. The marked block is $I_{(m)}$ for both $C$ and $\tilde{C}$, so $\tilde{\mathcal{U}}_M = \mathcal{U}_M$ exactly. Therefore $\|\sin\Theta(\mathcal{U}_M, \tilde{\mathcal{U}}_M)\|_2 = 0$ and the rotation of $\mathcal{U}_{\mathrm{eff}} = \mathcal{U}_M \oplus
		\mathcal{U}_\star$ is determined entirely by the rotation of $\mathcal{U}_\star$.
		
		\medskip
		\noindent\textbf{\textit{Gap associated with $\mathcal{U}_\star$.}}
		Between $\lambda_\star$ and the other eigenvalues $1$ and $-\tfrac{1}{n-1}$, the smallest gap is
		\begin{equation}
			\delta_\star = \min\!\left\{\frac{n-m}{n-1},\,
			\frac{m}{n-1}\right\}.
		\end{equation}
		We distinguish three main cases:
		\begin{itemize}
			\item \textbf{\textit{Few marked nodes:}} for $m = o(n)$, the minimum relevant gap is $\delta_\star = \tfrac{m}{n-1}$.
			\item \textbf{\textit{Marked nodes in an intermediate range:}} recalling	that $k = n-m$ is the number of unmarked nodes, there exists $\varepsilon \in (0,1/2]$ such that $\varepsilon n \le m \le
			(1-\varepsilon)n$ and $\varepsilon n \le k \le (1-\varepsilon)n$. Since both $\tfrac{n-m}{n-1}$ and $\tfrac{m}{n-1}$ are $\Theta(1)$, then also $\delta_\star = \Theta(1)$.
			\item \textbf{\textit{Few unmarked nodes:}} when $k = n-m = o(n)$, the minimum relevant gap is $\delta_\star = \tfrac{n-m}{n-1}$.
		\end{itemize}
		With $m = \tfrac{n-1}{a}$ and $a = 1.44512$,
		\begin{equation}
			\frac{m}{n-1} \approx \frac{1}{a} \approx 0.6918,
			\qquad
			\frac{n-m}{n-1} \approx 1 - \frac{1}{a} \approx 0.3080,
		\end{equation}
		thus the minimum gap is $\delta_\star = \tfrac{n-m}{n-1} = \Theta(1)$. With $\|E\|_2 = O(1/n)$, the Davis--Kahan bound gives
		\begin{equation}
			\sin\alpha_{\mathrm{eff}}
			= \|\sin\Theta(\mathcal{U}_\star,
			\tilde{\mathcal{U}}_\star)\|_2
			\leq \frac{\|E\|_2}{\delta_\star}
			= O\!\left(\frac{1}{n}\right),
		\end{equation}
		which establishes~\eqref{eq:subspace_rotation_bound}.\qedhere
	\end{proof}
	
	\begin{proof}[Proof of \Cref{lem:A_perturbation}]
		\medskip\noindent\textbf{Spectral norm of the isometry perturbation.}
		Since $A = \sum_u|\psi_u\rangle\langle u|$ and the perturbation $\Delta$ modifies only columns $u$ and $v$ of $P$, only $|\psi_u\rangle$ and $|\psi_v\rangle$ change, so:
		\begin{equation}
			\tilde{A} - A =
			|\delta_u\rangle\langle u|
			+ |\delta_v\rangle\langle v|,
		\end{equation}
		where $|\delta_u\rangle := |\tilde\psi_u\rangle - |\psi_u\rangle$ and $|\delta_v\rangle := |\tilde\psi_v\rangle - |\psi_v\rangle$. Since $\tilde{A} - A$ is a linear operator from $\mathbb{R}^n$ to
		$\mathcal{H}$, its spectral norm satisfies $\|\tilde{A} - A\|_2^2 = \lambda_{\max}\bigl((\tilde{A} - A)^\dagger (\tilde{A} - A)\bigr)=  \lambda_{\max}\bigl((\tilde{A} - A)^* (\tilde{A} - A)\bigr)$. Computing $(\tilde{A} - A)^* (\tilde{A} - A)$ as an operator on $\mathbb{R}^n$:
		\begin{align}
			(\tilde{A} - A)^*(\tilde{A} - A)
			&= (|u\rangle\langle\delta_u| + |v\rangle\langle\delta_v|)
			(|\delta_u\rangle\langle u| + |\delta_v\rangle\langle v|)
			\notag \\
			&= \|\delta_u\|_2^2\,|u\rangle\langle u|
			+ \langle\delta_u|\delta_v\rangle\,|u\rangle\langle v|
			+ \langle\delta_v|\delta_u\rangle\,|v\rangle\langle u|
			+ \|\delta_v\|_2^2\,|v\rangle\langle v|.
		\end{align}
		
		Since $|\psi_u\rangle = |u\rangle \otimes \sum_w\sqrt{P_{w,u}}|w\rangle$ and $|\tilde\psi_u\rangle = |u\rangle \otimes \sum_w\sqrt{\tilde{P}_{w,u}}|w\rangle$, their difference has the form
		$|\delta_u\rangle = |u\rangle \otimes |\eta_u\rangle$ where $|\eta_u\rangle = \sum_w(\sqrt{\tilde{P}_{w,u}} - \sqrt{P_{w,u}})|w\rangle$, and similarly $|\delta_v\rangle = |v\rangle \otimes |\eta_v\rangle$. Therefore:
		\begin{equation}
			\langle\delta_u|\delta_v\rangle
			= \langle u|v\rangle\,\langle\eta_u|\eta_v\rangle
			= 0,
		\end{equation}
		since $|u\rangle$ and $|v\rangle$ are distinct computational basis states and thus orthogonal. The operator $(\tilde{A}-A)^*(\tilde{A}-A)$ therefore reduces to:
		\begin{equation}
			(\tilde{A} - A)^*(\tilde{A} - A)
			= \|\delta_u\|_2^2\,|u\rangle\langle u|
			+ \|\delta_v\|_2^2\,|v\rangle\langle v|,
		\end{equation}
		which is diagonal in the computational basis with eigenvalues $\|\delta_u\|_2^2$ and $\|\delta_v\|_2^2$ (and zero on the orthogonal complement of $\mathrm{span}\{|u\rangle,|v\rangle\}$). Thus we have
		\begin{equation}
			\|\tilde{A} - A\|_2
			= \max\bigl(\||\delta_u\rangle\|_2,\;
			\||\delta_v\rangle\|_2\bigr).
		\end{equation}
		
		We compute $\||\delta_u\rangle\|_2^2 = \sum_w(\sqrt{\tilde{P}_{w,u}} - \sqrt{P_{w,u}})^2$. For $w = v$, since the removal of the link $u \to v$ sets $\sqrt{\tilde{P}_{v,u}} = 0$ from $\sqrt{P_{v,u}} = 1/\sqrt{n-1}$, we have:
		\begin{equation}
			\left(\sqrt{\tilde{P}_{v,u}} - \sqrt{P_{v,u}}\right)^2
			= \frac{1}{n-1}.
		\end{equation}
		For each of the remaining $n-2$ terms with $w \neq u,v$ we have the redistributed weight for $\tilde{P}$ which gives:
		\begin{equation}
			\left(\frac{1}{\sqrt{n-2}} - \frac{1}{\sqrt{n-1}}\right)^2
			= \frac{(\sqrt{n-1} - \sqrt{n-2})^2}{(n-1)(n-2)}
			= \frac{1}{(n-1)(n-2)(\sqrt{n-1}+\sqrt{n-2})^2}
			= O\!\left(\frac{1}{n^3}\right),
		\end{equation}
		so their total contribution over $n-2$ terms is $O(1/n^2)$. Therefore:
		\begin{equation}
			\||\delta_u\rangle\|_2^2
			= \frac{1}{n-1} + O\!\left(\frac{1}{n^2}\right)
			= \frac{1}{n-1}\left(1 + O\!\left(\frac{1}{n}\right)\right).
		\end{equation}
		Taking the square root and using $\sqrt{1+\varepsilon} = 1 + O(\varepsilon)$ for small $\varepsilon$:
		\begin{equation}
			\||\delta_u\rangle\|_2
			= \frac{1}{\sqrt{n-1}}\left(1 + O\!\left(\frac{1}{n}\right)\right)
			= \frac{1}{\sqrt{n-1}} + O\!\left(\frac{1}{n^{3/2}}\right)
			= \Theta\!\left(\frac{1}{\sqrt{n}}\right).
		\end{equation}
		The identical argument applies to $\||\delta_v\rangle\|_2$ by symmetry. Therefore:
		\begin{equation}
			\|\tilde{A} - A\|_2 = \Theta\!\left(\frac{1}{\sqrt{n}}\right).
		\end{equation}
		The same holds for $\|\tilde{B} - B\|_2$, since $B = \sum_v|\phi_v\rangle\langle v|$ and only $|\phi_u\rangle$ and $|\phi_v\rangle$ change under the perturbation, by the same argument.\qedhere
	\end{proof}

	\begin{proof}[Proof of \Cref{prop:omega_star_rotation}]
		\medskip\noindent\textbf{Rotation of $| \omega_\star^\pm \rangle$ in $\mathcal{H}$.}
		Using the explicit expressions from \cite{Portugal_book}(Theorem~11.4), we write:
		\begin{equation}\label{eq:omega_star_diff}
			\ket{\tilde\omega^\pm_\star} - \ket{\omega^\pm_\star}
			= \frac{\tilde{A}\tilde{\mathbf{v}}_\star
			- e^{\pm i\tilde\theta_\star}\tilde{B}\tilde{\mathbf{v}}_\star}{
			\sqrt{2}\sin\tilde\theta_\star}
			- \frac{A\mathbf{v}_\star
			- e^{\pm i\theta_\star}B\mathbf{v}_\star}{
			\sqrt{2}\sin\theta_\star}.
		\end{equation}
		Throughout this proof we work in the regime $m = \tfrac{n-1}{a}=\Theta(n)$, which ensures $\sin\theta_\star = \Theta(1)$ and, in this intermediate regime, $|\Delta\theta_\star| = O(1/n^2)$ (\Cref{prop:phase_shift_explicit}, intermediate-regime specialization). We add and subtract to \Cref{eq:omega_star_diff} $(\tilde{A}\tilde{\mathbf{v}}_\star - e^{\pm i\tilde\theta_\star}\tilde{B}\tilde{\mathbf{v}}_\star)/(\sqrt{2}\sin\theta_\star)$ and apply the triangle inequality:
		\begin{align}
			\|\ket{\tilde\omega^\pm_\star} - \ket{\omega^\pm_\star}\|_2
			\leq \frac{1}{\sqrt{2}\sin\theta_\star}
			\|\tilde{A}\tilde{\mathbf{v}}_\star
			- e^{\pm i\tilde\theta_\star}\tilde{B}\tilde{\mathbf{v}}_\star
			- A\mathbf{v}_\star+ e^{\pm i\theta_\star}B\mathbf{v}_\star\|_2 
			+ \left|\frac{1}{\sqrt{2}\sin\tilde\theta_\star}
			- \frac{1}{\sqrt{2}\sin\theta_\star}\right|
			\cdot\|\tilde{A}\tilde{\mathbf{v}}_\star
			- e^{\pm i\tilde\theta_\star}
			\tilde{B}\tilde{\mathbf{v}}_\star\|_2.
		\end{align}
		
		\medskip
		\noindent\textbf{\textit{Second term.}}
		Since $\|\tilde{A}\tilde{\mathbf{v}}_\star\|_2 \leq \|\tilde{A}\|_2 = 1$ and $\|\tilde{B}\tilde{\mathbf{v}}_\star\|_2 \leq 1$, we have $\|\tilde{A}\tilde{\mathbf{v}}_\star - e^{\pm i\tilde\theta_\star}\tilde{B}\tilde{\mathbf{v}}_\star\|_2 \leq 2$. The prefactor satisfies:
		\begin{equation}
			\left|\frac{1}{\sin\tilde\theta_\star}
			- \frac{1}{\sin\theta_\star}\right|
			= \frac{|\sin\theta_\star - \sin\tilde\theta_\star|}{
				\sin\theta_\star\sin\tilde\theta_\star}
			\leq \frac{|\Delta\theta_\star|}{\sin^2\theta_\star}
			= O\!\left(\frac{1}{n^2}\right),
		\end{equation}
		where we used $|\Delta\theta_\star| = O(1/n^2)$ (\Cref{prop:phase_shift_explicit}, specialized to $m=\Theta(n)$) and $\sin\theta_\star = \Theta(1)$. Therefore this term contributes $O(1/n^2)$ and is negligible.
		
		\medskip
		\noindent\textbf{\textit{First term.}}
		We rewrite the numerator by grouping the $A$ and $B$ terms separately:
		\begin{equation}
			\tilde{A}\tilde{\mathbf{v}}_\star
			- e^{\pm i\tilde\theta_\star}\tilde{B}\tilde{\mathbf{v}}_\star
			- A\mathbf{v}_\star + e^{\pm i\theta_\star}B\mathbf{v}_\star
			 = \bigl(\tilde{A}\tilde{\mathbf{v}}_\star
			- A\mathbf{v}_\star\bigr)
			- \bigl(e^{\pm i\tilde\theta_\star}
			\tilde{B}\tilde{\mathbf{v}}_\star
			- e^{\pm i\theta_\star}B\mathbf{v}_\star\bigr).
		\end{equation}
		For the first group, adding and subtracting	$A\tilde{\mathbf{v}}_\star$, we get:
		\begin{equation}
			\tilde{A}\tilde{\mathbf{v}}_\star - A\mathbf{v}_\star
			= (\tilde{A}-A)\tilde{\mathbf{v}}_\star
			+ A(\tilde{\mathbf{v}}_\star - \mathbf{v}_\star).
		\end{equation}
		For the second group, adding and subtracting $e^{\pm i\tilde\theta_\star}B\mathbf{v}_\star$:
		\begin{align}
			e^{\pm i\tilde\theta_\star}\tilde{B}\tilde{\mathbf{v}}_\star
			- e^{\pm i\theta_\star}B\mathbf{v}_\star
			&= e^{\pm i\tilde\theta_\star}
			\bigl(\tilde{B}\tilde{\mathbf{v}}_\star
			- B\mathbf{v}_\star\bigr)
			+ \bigl(e^{\pm i\tilde\theta_\star}
			- e^{\pm i\theta_\star}\bigr)B\mathbf{v}_\star \notag \\
			&= e^{\pm i\tilde\theta_\star}\bigl[
			(\tilde{B}-B)\tilde{\mathbf{v}}_\star
			+ B(\tilde{\mathbf{v}}_\star - \mathbf{v}_\star)\bigr]
			+ \bigl(e^{\pm i\tilde\theta_\star}
			- e^{\pm i\theta_\star}\bigr)B\mathbf{v}_\star,
		\end{align}
		where in the second line we further decomposed $\tilde{B}\tilde{\mathbf{v}}_\star - B\mathbf{v}_\star$ by adding and subtracting $B\tilde{\mathbf{v}}_\star$. Applying the triangle inequality to the full expression, and using $|e^{\pm i\tilde\theta_\star}| = 1$ and $\|B\mathbf{v}_\star\|_2 \leq \|B\|_2 = 1$:
		\begin{align}
			\|\tilde{A}\tilde{\mathbf{v}}_\star
			- e^{\pm i\tilde\theta_\star}\tilde{B}\tilde{\mathbf{v}}_\star
			- A\mathbf{v}_\star + e^{\pm i\theta_\star}
			B\mathbf{v}_\star\|_2
			&\leq \|(\tilde{A}-A)\tilde{\mathbf{v}}_\star\|_2 + \|A(\tilde{\mathbf{v}}_\star- \mathbf{v}_\star)\|_2 
			 + \|(\tilde{B}-B)\tilde{\mathbf{v}}_\star\|_2  \notag \\
			&\quad + \|B(\tilde{\mathbf{v}}_\star- \mathbf{v}_\star)\|_2 + |e^{\pm i\tilde\theta_\star} - e^{\pm i\theta_\star}|.
		\end{align}
		Using $|e^{i\alpha} - e^{i\beta}| \leq |\alpha - \beta|$, the last term is bounded by $|\Delta\theta_\star| = O(1/n^2)$ and is negligible. We bound the remaining four terms separately.
		
		\medskip
		\noindent\textbf{\textit{Contribution from $A \to \tilde{A}$ and $B \to \tilde{B}$.}}
		Since $\|\tilde{\mathbf{v}}_\star\|_2 = 1$:
		\begin{equation}
			\|(\tilde{A}-A)\tilde{\mathbf{v}}_\star\|_2
			\leq \|\tilde{A}-A\|_2
			= \Theta\!\left(\frac{1}{\sqrt{n}}\right),
		\end{equation}
		and the same holds for $\|(\tilde{B}-B)\tilde{\mathbf{v}}_\star\|_2$, by \Cref{lem:A_perturbation}.
		
		\medskip
		\noindent\textbf{\textit{Contribution from $\mathbf{v}_\star \to \tilde{\mathbf{v}}_\star$.}}
		Since $\|A\|_2 = \|B\|_2 = 1$ and $\|\tilde{\mathbf{v}}_\star - \mathbf{v}_\star\|_2 \leq \|E\|_2/\delta_\star = O(1/n)$ from \Cref{prop:subspace_rotation}:
		\begin{equation}
			\|A(\tilde{\mathbf{v}}_\star - \mathbf{v}_\star)\|_2
			\leq \|\tilde{\mathbf{v}}_\star - \mathbf{v}_\star\|_2
			= O\!\left(\frac{1}{n}\right),
		\end{equation}
		and the same holds for $\|B(\tilde{\mathbf{v}}_\star - \mathbf{v}_\star)\|_2$.
		
		Since $\sin\theta_\star = \Theta(1)$, collecting all contributions:
		\begin{equation}
			\|\ket{\tilde\omega^\pm_\star} - \ket{\omega^\pm_\star}\|_2
			= \frac{1}{\Theta(1)}\left(O\!\left(\frac{1}{\sqrt{n}}\right)
			+ O\!\left(\frac{1}{n}\right)\right)
			+ O\!\left(\frac{1}{n^2}\right)
			= O\!\left(\frac{1}{\sqrt{n}}\right),
		\end{equation}
		dominated by the change of the isometries $A \to \tilde{A}$ and $B \to \tilde{B}$ defined by \Cref{lem:A_perturbation}.\qedhere
	\end{proof}
	\begin{proof}[Proof of \Cref{prop:V_star_rotation}]
		\medskip\noindent\textbf{Rotation of $\mathcal{V}_\star$ in $\mathcal{H}$.}
		The eigenvalues of $W$ in $\mathcal{V}_\star$ are $e^{\pm 2i\theta_\star}$. The other eigenvalues of $W$ are $1$ for $\mathcal{V}_M$ and $\mathcal{A}^\perp \cap \mathcal{B}^\perp$, and $e^{\pm 2i\theta_B}$ for $\mathcal{V}_B$. We apply \eqref{eq:bhatia-projector-bound} from Theorem VII.3.4 of \cite{Bhatia1997} with $N_1 = W$, $N_2 = \tilde{W}$, $S_1 = \{e^{\pm 2i\theta_\star}\}$, and
		$S_2 = \sigma(\tilde{W}|_{\tilde{\mathcal{V}}_\star^\perp})$. The separation is:
		\begin{equation}
			\delta_W = \mathrm{dist}(S_1, S_2)
			= \min\bigl(|e^{\pm 2i\theta_\star} - 1|,\;
			|e^{\pm 2i\theta_\star} - e^{\pm 2i\theta_B}|\bigr)
		\end{equation}
		which is $\delta_W= \Theta(1)$ for $m = \Theta(n)$, since both $\theta_\star$ and $\theta_B$ are bounded away from $0$ and $\pi/2$. Therefore by \Cref{eq:bhatia-projector-bound} and \Cref{cor:W_perturbation}, within the regime $m=\Theta(n)$, we have:
		\begin{equation}
			\|\Pi_{\tilde{\mathcal{V}}_\star}
			- \Pi_{\mathcal{V}_\star}\|_2
			\leq \frac{\|\tilde{W} - W\|_2}{\delta_W}
			= \frac{\Theta(1/\sqrt{n})}{\Theta(1)}
			= O\!\left(\frac{1}{\sqrt{n}}\right).\qedhere
		\end{equation}	
	\end{proof}
	
	\begin{proof}[Proof of \Cref{prop:bulk_orthogonality}]
		\medskip\noindent\textbf{Bulk orthogonality.}
		The original discriminant matrix $C$ is block-diagonal, $C = I_{(m)} \oplus C_U$, reflecting the separation between $M$ and $U$. Since the perturbation arises from the removal of an edge
		connecting two unmarked nodes $u, v \in U$, the matrix $E$ has non-zero entries strictly confined to $U$, i.e.\ $E = 0_{(m)} \oplus E_U$, so that
		\begin{equation}
			\tilde{C} = \begin{pmatrix} I_{(m)} & 0 \\ 0 & C_U + E_U
			\end{pmatrix}.
		\end{equation}
		By the spectral theory of block-diagonal matrices, the perturbed bulk eigenvectors in $\tilde{\mathcal{U}}_B$ and the perturbed gap eigenvector in $\tilde{\mathcal{U}}_\star$ retain support strictly on the unmarked block, while the marked eigenvectors in $\tilde{\mathcal{U}}_M$ remain supported strictly on $M$. We denote by $\{|\tilde\omega_{B,l}\rangle\}_l$ the orthonormal basis of
		eigenvectors of $\tilde{C}$ associated with $\tilde{\mathcal{U}}_B$. Since each $|\tilde\omega_{B,l}\rangle$ has support strictly on the unmarked block,
		\begin{equation}
			\Pi_M|\tilde\omega_{B,l}\rangle = 0 \quad \text{for all } l,
		\end{equation}
		so the bulk eigenvectors do not contribute to the success probability upon measurement. Consequently, for any state $|\psi\rangle \in \tilde{\mathcal{V}}_B$, we have $\Pi_M\tilde{W}^t|\psi\rangle = 0$ exactly, since $\tilde{W}^t$ preserves $\tilde{\mathcal{V}}_B$.\qedhere
	\end{proof}

}
%FIFTH APPENDIX END----------------------------------------------------------------------------

\section{Effect on the Success Probability of a Single Edge Removal}
\label{sec:effects_on_succ_prob}

Having evaluated the spectral perturbation induced by a single undirected edge removal, we now assess its impact on the success probability of the Szegedy walk. Recall that for a marked set $M$,
the success probability at discrete time $t$ is defined as~\cite{Portugal_book}
\begin{equation}
	p_M(t) \;=\; \|\Pi_M\, W^t\!\ket{\psi_0}\|^2,
\end{equation}
where $W$ is Szegedy's walk operator, $\Pi_M$ projects onto the marked subspace, and $\ket{\psi_0}$ is the uniform initial state.

We saw that the variation of $p_M(t)$ under a small perturbation of $W$ originates from two sources: the phase shift of the eigenvalue $\theta_\star \to \tilde{\theta}_\star$, and the geometric misalignment between $\mathcal{V}_{\mathrm{eff}}$ and $\tilde{\mathcal{V}}_{\mathrm{eff}}$ in $\mathcal{H}$. We now combine these two effects and analyze their joint impact.

\subsection{Combined effect on the success probability}

We compute the overall variation of the success probability combining the phase shift $\Delta\theta_\star$ associated with the change $\lambda_\star \to \tilde{\lambda}_\star$, and the rotation of
$\mathcal{V}_{\mathrm{eff}}$ in $\mathcal{H}$.

\begin{proposition}[Perturbation bound on the success probability]
	\label{prop:prob_perturbation_direct}
	Let $p_M(t)$ and $\tilde{p}_M(t)$ be the success probabilities of the Szegedy walk on $K_n$ and $G_{u,v}$ respectively, with $m = \Theta(n)$ marked nodes and uniform initial state $\ket{\psi_0}$. For any $t \geq 1$,
	\begin{equation}
		|\tilde{p}_M(t) - p_M(t)|
		\;=\; O\!\left(\frac{1}{\sqrt{n}}\right)
		+ O\!\left(\frac{t}{n^2}\right),
	\end{equation}
	and in particular for $t = t^\star = O(1)$,
	\begin{equation}
		|\tilde{p}_M(t^\star) - p_M(t^\star)|
		\;=\; O\!\left(\frac{1}{\sqrt{n}}\right),
	\end{equation}
	dominated by the geometric misalignment between $\mathcal{V}_{\mathrm{eff}}$ and $\tilde{\mathcal{V}}_{\mathrm{eff}}$, controlled by $\|\Pi_{\tilde{\mathcal{V}}_{\mathrm{eff}}} - 	\Pi_{\mathcal{V}_{\mathrm{eff}}}\|_2 = O(1/\sqrt{n})$ (\Cref{cor:W_perturbation}, Theorem~VII.3.4 of~\cite{Bhatia1997}).
\end{proposition}

The proof is given in \nameshortref{app:prob_perturbation}. We summarize the main steps here. As established in \Cref{subsec:discriminant_with_marked_vertices_block_form}, the component of $\ket{\psi_0}$ outside $\mathcal{V}_{\mathrm{eff}}$ lies in $\mathcal{A}^\perp \cap \mathcal{B}^\perp$ and has zero overlap with the marked subspace, while $\ket{\psi_0} \perp \mathcal{V}_B$ exactly. Therefore both $p_M(t) = \|\Pi_M W^t \Pi_{\mathcal{V}_{\mathrm{eff}}}\ket{\psi_0}\|^2$ and $\tilde{p}_M(t) = \|\Pi_M \tilde{W}^t \Pi_{\tilde{\mathcal{V}}_{\mathrm{eff}}}\ket{\psi_0}\|^2$ hold exactly, with no remainder terms. Defining $|\psi_0^{\parallel} \rangle = \Pi_{\mathcal{V}_{\mathrm{eff}}}\ket{\psi_0}$ and $| \tilde{\psi}_0^{\parallel} \rangle = \Pi_{\tilde{\mathcal{V}}_{\mathrm{eff}}}\ket{\psi_0}$, the total variation reduces exactly to $|\tilde{p}_M^\parallel(t) - p_M^\parallel(t)|$, which we bound through two contributions (notice that for $| \tilde{\psi}_0^{\parallel} \rangle = \Pi_{\tilde{\mathcal{V}}_{\mathrm{eff}}}\ket{\psi_0}$ the tilde denotes the projection on the perturbed subspace, while the state $\ket{\psi_0}$ is unperturbed). The geometric contribution bounds the difference between $\|\Pi_M W^t | \tilde{\psi}^{\parallel}_0 \rangle \|^2$ and $\|\Pi_M W^t|\psi_0^{\parallel} \rangle\|^2$, arising from the misalignment $| \tilde{\psi}_0^{\parallel} \rangle - |\psi_0^{\parallel} \rangle = (\Pi_{\tilde{\mathcal{V}}_{\mathrm{eff}}} - \Pi_{\mathcal{V}_{\mathrm{eff}}})\ket{\psi_0}$. Since $\tilde{\mathcal{V}}_M = \mathcal{V}_M$ exactly (\Cref{prop:subspace_rotation}), the projector difference reduces to $\Pi_{\tilde{\mathcal{V}}_{\mathrm{eff}}} - \Pi_{\mathcal{V}_{\mathrm{eff}}} = \Pi_{\tilde{\mathcal{V}}_\star} - \Pi_{\mathcal{V}_\star}$. By \Cref{prop:V_star_rotation} we have:
\begin{equation}
	\|\Pi_{\tilde{\mathcal{V}}_{\mathrm{eff}}}
	- \Pi_{\mathcal{V}_{\mathrm{eff}}}\|_2
	= \|\Pi_{\tilde{\mathcal{V}}_\star}
	- \Pi_{\mathcal{V}_\star}\|_2
	\leq O\!\left(\frac{1}{\sqrt{n}}\right),
\end{equation}
This gives a geometric contribution of $O(1/\sqrt{n})$. The spectral contribution bounds the difference between $\|\Pi_M\tilde{W}^t| \tilde{\psi}_0^{\parallel} \rangle\|^2$ and $\|\Pi_M W^t| \tilde{\psi}_0^{\parallel} \rangle\|^2$, arising from the change in both the eigenvectors $\ket{\omega^\pm_\star} \to \ket{\tilde\omega^\pm_\star}$ and the eigenphase $\theta_\star \to \tilde\theta_\star$. By \Cref{prop:omega_star_rotation} we have $\|\ket{\tilde\omega^\pm_\star} - \ket{\omega^\pm_\star}\|_2= O\!\left(1/\sqrt{n}\right),$
dominated by the change of the isometries $A \to \tilde{A}$ and $B \to \tilde{B}$, while the eigenphase shift contributes $O(t|\Delta\theta_\star|) = O(t/n^2)$, using the intermediate-regime specialization $\Delta\theta_\star=\Theta(1/n^2)$ for $m=\Theta(n)$ established in \Cref{prop:phase_shift_explicit}. This gives a spectral contribution of $O(1/\sqrt{n}) + O(t/n^2)$. Combining the two contributions, for any $t \geq 1$:
\begin{equation}
	|\tilde{p}_M(t) - p_M(t)|
	= O\!\left(\frac{1}{\sqrt{n}}\right)
	+ O\!\left(\frac{t}{n^2}\right),
\end{equation}
and in particular for $t = O(1)$, which comes from the regime $m=\Theta(n)$ \cite{GIORDANO2026170305}, the term $O(t/n^2)$ is negligible and the bound reduces to $O(1/\sqrt{n})$, dominated by the geometric misalignment of $\mathcal{V}_{\mathrm{eff}}$.

\paragraph{Physical interpretation.} This result identifies which of the two spectral effects actually limits the monitoring procedure. One might expect the loss of success probability to be governed by the change in the search frequency, i.e.\ by the eigenphase shift $\Delta\theta_\star=\Theta(1/n^2)$. The bound shows instead that the dominant mechanism is geometric: the missing link tilts the effective invariant subspace $\mathcal V_{\mathrm{eff}}$ by a principal angle of order $1/\sqrt n$, and it is this misalignment between the ideal and the perturbed subspaces, not the slower precession, that leaks amplitude away from the marked set. The two effects scale differently: the phase term contributes only $O(t/n^2)$ (and stays negligible over the $O(1)$ search time, used in the completeness testing algorithm \cite{GIORDANO2026170305}) while the subspace tilt contributes $O(1/\sqrt n)$ and dominates. Physically, detecting a single failed link is therefore a problem of subspace tracking rather than of frequency resolution: monitoring how the effective subspace rotates is a more sensitive probe of the anomaly than estimating the eigenphase, and the $O(1/\sqrt n)$ law sets the fundamental precision that any such quantum monitoring scheme must reach as the network grows.

%SIXTH APPENDIX BEGIN-----------------------------------------------------------------
\toappendix{
	\stepcounter{AppCount}
	\subsection*{\Alph{AppCount}.\space Success probability bound}
	\labelshort[\Alph{AppCount}]{app:prob_perturbation}
	In order to prove \Cref{prop:prob_perturbation_direct}, we have to provide a corollary to \Cref{lem:A_perturbation} which will be used in the proof.
	
	\begin{corollary}[Spectral norm of the walk operator perturbation]
		\label{cor:W_perturbation}
		Let $W = R_B R_A$ and $\tilde{W} = \tilde{R}_B\tilde{R}_A$ be the Szegedy walk operators associated with the complete graph $K_n$ with marked nodes and the graph $G_{u,v}$ obtained by removing the edge $\{u,v\}$ with uniform redistribution, respectively. Then:
		\begin{equation}
			\|\tilde{W} - W\|_2 = \Theta\!\left(\frac{1}{\sqrt{n}}\right).
		\end{equation}
	\end{corollary}
	
	\begin{proof}[Proof of \Cref{cor:W_perturbation}]
		\medskip\noindent\textbf{Spectral norm of the walk operator perturbation.}
		From $W = R_B R_A$ and $\tilde{W} = \tilde{R}_B\tilde{R}_A$, adding and subtracting $\tilde{R}_B R_A$ and applying the triangle inequality:
		\begin{equation}
			\|\tilde{W} - W\|_2
			\leq \|\tilde{R}_B(\tilde{R}_A - R_A)\|_2
			+ \|(\tilde{R}_B - R_B)R_A\|_2
			\leq \|\tilde{R}_A - R_A\|_2 + \|\tilde{R}_B - R_B\|_2,
		\end{equation}
		where we used $\|\tilde{R}_B\|_2 = \|R_A\|_2 = 1$ since reflections are unitary. Since $R_A = 2AA^T - I$, we have:
		\begin{equation}
			\tilde{R}_A - R_A = 2(\tilde{A}\tilde{A}^T - AA^T)
			= 2\bigl[(\tilde{A}-A)\tilde{A}^T + A(\tilde{A}-A)^T\bigr],
		\end{equation}
		so using $\|\tilde{A}\|_2 = \|A\|_2 = 1$ (both are isometries):
		\begin{equation}
			\|\tilde{R}_A - R_A\|_2
			\leq 2\bigl(\|\tilde{A}-A\|_2\|\tilde{A}\|_2
			+ \|A\|_2\|\tilde{A}-A\|_2\bigr)
			= 4\|\tilde{A} - A\|_2
			= \Theta\!\left(\frac{1}{\sqrt{n}}\right),
		\end{equation}
		by \Cref{lem:A_perturbation}. The identical bound holds for $\|\tilde{R}_B - R_B\|_2$. Combining:
		\begin{equation}
			\|\tilde{W} - W\|_2 = \Theta\!\left(\frac{1}{\sqrt{n}}\right).\qedhere
		\end{equation}
	\end{proof}
	
	\begin{proof}[Proof of \Cref{prop:prob_perturbation_direct}]
		\medskip\noindent\textbf{Perturbation bound on the success probability.}
		
		\medskip
		\noindent\textbf{\textit{Decomposition of $\ket{\psi_0}$.}}
		Since $\ket{\psi_0} \perp \mathcal{V}_B$ exactly (\Cref{subsec:discriminant_with_marked_vertices_block_form}), the state $\ket{\psi_0}$ decomposes as:
		\begin{equation}
			\ket{\psi_0}
			= \underbrace{\Pi_{\mathcal{V}_{\mathrm{eff}}}\ket{\psi_0}}_{
				|\psi_0^\parallel\rangle }
			+ \underbrace{\Pi_{\mathcal{A}^\perp \cap \mathcal{B}^\perp}
				\ket{\psi_0}}_{|\psi_0^\perp\rangle},
		\end{equation}
		where $|\psi_0^\parallel \rangle \in \mathcal{V}_{\mathrm{eff}}$ and $|\psi_0^\perp\rangle \in \mathcal{A}^\perp \cap \mathcal{B}^\perp$, with $\||\psi_0^{\parallel} \rangle\|_2^2 + \|\ket{\psi_0^\perp}\|_2^2 = 1$. Similarly, we define the projection of $\ket{\psi_0}$ onto the perturbed effective subspace:
		\begin{equation}
			| \tilde{\psi}_0^{\parallel} \rangle
			:= \Pi_{\tilde{\mathcal{V}}_{\mathrm{eff}}}\ket{\psi_0}.
		\end{equation}
		
		\medskip
		\noindent\textbf{\textit{Reduction to the effective subspace.}}
		We define:
		\begin{equation}
			p_M^\parallel(t) := \|\Pi_M W^t|\psi_0^{\parallel} \rangle\|^2,
			\qquad
			\tilde{p}_M^\parallel(t) :=
			\|\Pi_M\tilde{W}^t| \tilde{\psi}_0^{\parallel} \rangle\|^2.
		\end{equation}
		We claim that $p_M(t) = p_M^\parallel(t)$ and $\tilde{p}_M(t) = \tilde{p}_M^\parallel(t)$ exactly. For the unperturbed case: since $\ket{\psi_0^\perp} \in \mathcal{A}^\perp \cap \mathcal{B}^\perp$, which is a $(+1)$-eigenspace of $W$ satisfying $\Pi_M(\mathcal{A}^\perp \cap \mathcal{B}^\perp) = 0$ (see \Cref{subsec:discriminant_with_marked_vertices_block_form}), we have $\Pi_M W^t\ket{\psi_0^\perp} = \Pi_M\ket{\psi_0^\perp} = 0$ exactly,
		so:
		\begin{equation}
			p_M(t) = \|\Pi_M W^t\ket{\psi_0}\|^2
			= \|\Pi_M W^t|\psi_0^{\parallel} \rangle\|^2
			= p_M^\parallel(t).
		\end{equation}
		
		For the perturbed case: decomposing $\ket{\psi_0}$ in the eigenspaces of $\tilde{W}$, the component in $\tilde{\mathcal{V}}_B$ is annihilated exactly by $\Pi_M$ via \Cref{prop:bulk_orthogonality}, and the component in $\tilde{\mathcal{A}}^\perp \cap \tilde{\mathcal{B}}^\perp$ satisfies $\Pi_M(\tilde{\mathcal{A}}^\perp \cap \tilde{\mathcal{B}}^\perp) = 0$
		since the marked block is unchanged under the perturbation. Therefore:
		\begin{equation}
			\tilde{p}_M(t) = \|\Pi_M\tilde{W}^t\ket{\psi_0}\|^2
			= \|\Pi_M\tilde{W}^t| \tilde{\psi}_0^{\parallel} \rangle\|^2
			= \tilde{p}_M^\parallel(t).
		\end{equation}
		
		\medskip
		\noindent\textbf{\textit{Analyzing the probability variation.}}
		Since $p_M(t) = p_M^\parallel(t)$ and $\tilde{p}_M(t) =	\tilde{p}_M^\parallel(t)$ exactly, the total variation reduces to:
		\begin{equation}\label{eq:triangle_split}
			|\tilde{p}_M(t) - p_M(t)|
			= |\tilde{p}_M^\parallel(t) - p_M^\parallel(t)|
			= \bigl|\|\Pi_M\tilde{W}^t| \tilde{\psi}_0^{\parallel} \rangle\|^2
			- \|\Pi_M W^t|\psi_0^{\parallel} \rangle\|^2\bigr|.
		\end{equation}
		Inserting and subtracting $\|\Pi_M W^t| \tilde{\psi}_0^{\parallel} \rangle\|^2$ and applying the triangle inequality we obtain:
		\begin{equation}
			|\tilde{p}_M(t) - p_M(t)|
			\leq
			\underbrace{\bigl|\|\Pi_M\tilde{W}^t| \tilde{\psi}_0^{\parallel} \rangle\|^2
				- \|\Pi_M W^t| \tilde{\psi}_0^{\parallel} \rangle\|^2\bigr|}_{\text{spectral}}
			+
			\underbrace{\bigl|\|\Pi_M W^t| \tilde{\psi}_0^{\parallel} \rangle\|^2
				- \|\Pi_M W^t|\psi_0^{\parallel} \rangle\|^2\bigr|}_{\text{geometric}},
		\end{equation}
		calling the first and second term respectively spectral contribution and geometric contribution to the probability variation. The first can be seen as the effect of the changes in the spectra of $W$ due to the perturbation, and the second one as the changes caused by the rotation of the eigenspace $\mathcal{V}_{\mathrm{eff}}$ due to the perturbation.
		
		\medskip
		\noindent\textbf{\textit{Geometric contribution.}}
		Using $|a^2-b^2| \leq 2|a-b|$, the reverse triangle inequality, and $\|\Pi_M\|_2 = \|W^t\|_2 = 1$:
		\begin{equation}
			\bigl|\|\Pi_M W^t|\tilde{\psi}_0^{\parallel}\rangle\|^2
			- \|\Pi_M W^t|\psi_0^{\parallel}\rangle\|^2\bigr|
			\leq 2\|\Pi_M W^t(|\tilde{\psi}_0^{\parallel}\rangle
			- |\psi_0^{\parallel}\rangle)\|_2
			\leq 2\||\tilde{\psi}_0^{\parallel}\rangle
			- |\psi_0^{\parallel}\rangle\|_2.
		\end{equation}
		Since $|\tilde{\psi}_0^{\parallel}\rangle - |\psi_0^{\parallel}\rangle = (\Pi_{\tilde{\mathcal{V}}_{\mathrm{eff}}} - \Pi_{\mathcal{V}_{\mathrm{eff}}})\ket{\psi_0}$:
		\begin{equation}
			\||\tilde{\psi}_0^{\parallel}\rangle
			- |\psi_0^{\parallel}\rangle\|_2
			\leq \|\Pi_{\tilde{\mathcal{V}}_{\mathrm{eff}}}
			- \Pi_{\mathcal{V}_{\mathrm{eff}}}\|_2.
		\end{equation}
		Since $\tilde{\mathcal{V}}_M = \mathcal{V}_M$ exactly (\Cref{prop:subspace_rotation}), we can decompose the projectors as $\Pi_{\mathcal{V}_{\mathrm{eff}}} = \Pi_{\mathcal{V}_M} + \Pi_{\mathcal{V}_\star}$ and $\Pi_{\tilde{\mathcal{V}}_{\mathrm{eff}}} = \Pi_{\mathcal{V}_M} + \Pi_{\tilde{\mathcal{V}}_\star}$, where the second decomposition uses $\tilde{\mathcal{V}}_\star \perp \tilde{\mathcal{V}}_M = \mathcal{V}_M$ since they are eigenspaces of $\tilde{W}$ for distinct eigenvalues. Therefore:
		\begin{equation}
			\|\Pi_{\tilde{\mathcal{V}}_{\mathrm{eff}}}
			- \Pi_{\mathcal{V}_{\mathrm{eff}}}\|_2
			= \|\Pi_{\tilde{\mathcal{V}}_\star}
			- \Pi_{\mathcal{V}_\star}\|_2
			= O\!\left(\frac{1}{\sqrt{n}}\right),
		\end{equation}
		where the last bound follows from \Cref{prop:V_star_rotation}. Therefore the geometric contribution is $O(1/\sqrt{n})$.
		
		\medskip
		\noindent\textbf{\textit{Spectral contribution.}}
		Using $|a^2-b^2| \leq 2|a-b|$, the reverse triangle inequality,	and $\|\Pi_M\|_2 = 1$:
		\begin{equation}\label{eq:spec_contrib_first_diff1}
			\bigl|\|\Pi_M\tilde{W}^t| \tilde{\psi}_0^{\parallel} \rangle\|^2
			- \|\Pi_M W^t| \tilde{\psi}_0^{\parallel} \rangle\|^2\bigr|
			\leq 2\|(\tilde{W}^t - W^t)| \tilde{\psi}_0^{\parallel} \rangle\|_2.
		\end{equation}
		We decompose $| \tilde{\psi}_0^{\parallel} \rangle \in \tilde{\mathcal{V}}_{\mathrm{eff}} = \tilde{\mathcal{V}}_\star \oplus \tilde{\mathcal{V}}_M$ in the eigenbasis of $\tilde{W}$:
		\begin{equation}
			| \tilde{\psi}_0^{\parallel} \rangle
			= \tilde{a}^+\ket{\tilde\omega^+_\star}
			+ \tilde{a}^-\ket{\tilde\omega^-_\star}
			+ \sum_{j=1}^m \tilde{c}_j\ket{\tilde\omega_{M,j}},
		\end{equation}
		where $|\tilde{a}^+|^2 + |\tilde{a}^-|^2 + \sum_j|\tilde{c}_j|^2 = \|| \tilde{\psi}_0^{\parallel} \rangle\|^2 \leq 1$. Since $\tilde{W}^t$ acts diagonally in this basis:
		\begin{equation}
			\tilde{W}^t| \tilde{\psi}_0^{\parallel} \rangle
			= \tilde{a}^+e^{2i\tilde\theta_\star t}\ket{\tilde\omega^+_\star}
			+ \tilde{a}^-e^{-2i\tilde\theta_\star t}\ket{\tilde\omega^-_\star}
			+ \sum_{j=1}^m \tilde{c}_j\ket{\tilde\omega_{M,j}}.
		\end{equation}
		Decomposing $| \tilde{\psi}_0^{\parallel} \rangle$ in the eigenbasis of $W$, the components in $\mathcal{V}_B$ and $\mathcal{A}^\perp \cap \mathcal{B}^\perp$ are annihilated by $\Pi_M$, so for the purpose of bounding $\|\Pi_M W^t| \tilde{\psi}_0^{\parallel} \rangle\|_2$ we write:
		\begin{equation}
			W^t| \tilde{\psi}_0^{\parallel} \rangle
			= a^+e^{2i\theta_\star t}\ket{\omega^+_\star}
			+ a^-e^{-2i\theta_\star t}\ket{\omega^-_\star}
			+ \sum_{j=1}^m c_j\ket{\omega_{M,j}}
			+ (\text{terms annihilated by }\Pi_M).
		\end{equation}
		Since the perturbation $E = 0_{(m)} \oplus E_U$ leaves the marked block unchanged, $\ket{\tilde\omega_{M,j}} = \ket{\omega_{M,j}}$ exactly (\Cref{prop:bulk_orthogonality}), with the same coefficients $c_j = \tilde{c}_j$, so the $\mathcal{V}_M$ terms cancel in $(\tilde{W}^t - W^t)| \tilde{\psi}_0^{\parallel} \rangle$. Expanding the difference and applying the triangle inequality for each $\pm$ term:
		\begin{equation}\label{eq:spec_contrib_first_diff3}
			\|\tilde{a}^\pm e^{\pm 2i\tilde\theta_\star t}
			\ket{\tilde\omega^\pm_\star}
			- a^\pm e^{\pm 2i\theta_\star t}\ket{\omega^\pm_\star}\|_2
			\leq |\tilde{a}^\pm|\,\|e^{\pm 2i\tilde\theta_\star t}
			\ket{\tilde\omega^\pm_\star}
			- e^{\pm 2i\theta_\star t}\ket{\omega^\pm_\star}\|_2
			+ |\tilde{a}^\pm - a^\pm|.
		\end{equation}
		For the first part, adding and subtracting $e^{\pm 2i\tilde\theta_\star t}\ket{\omega^\pm_\star}$:
		\begin{equation}
			\|e^{\pm 2i\tilde\theta_\star t}\ket{\tilde\omega^\pm_\star}
			- e^{\pm 2i\theta_\star t}\ket{\omega^\pm_\star}\|_2
			\leq \|\ket{\tilde\omega^\pm_\star} - \ket{\omega^\pm_\star}\|_2
			+ |e^{\pm 2i\tilde\theta_\star t} - e^{\pm 2i\theta_\star t}|.
		\end{equation}
		Since the two vectors $\ket{\tilde{\omega}_\star^\pm}$ and  $\ket{\omega_\star^\pm}$ belongs respectively to $\tilde{\mathcal{V}}_\star$ and $\mathcal{V}_\star$, by \Cref{prop:omega_star_rotation}:
		\begin{equation}
			\|\ket{\tilde\omega^\pm_\star} - \ket{\omega^\pm_\star}\|_2
			= O\!\left(\frac{1}{\sqrt{n}}\right),
		\end{equation}
		while for the eigenphase difference, $|e^{\pm 2i\tilde\theta_\star t} - e^{\pm 2i\theta_\star t}| \leq 2t|\Delta\theta_\star| = O(t/n^2)$ from \Cref{prop:phase_shift_explicit}.
		
		For the second part of the right side of \Cref{eq:spec_contrib_first_diff3}, since $\tilde{a}^\pm$ and $a^\pm$ are the coordinates of $| \tilde{\psi}_0^{\parallel} \rangle$ in orthonormal bases
		whose subspaces differ by $\|\Pi_{\tilde{\mathcal{V}}_\star} - \Pi_{\mathcal{V}_\star}\|_2 = O(1/\sqrt{n})$:
		\begin{equation}
			|\tilde{a}^\pm - a^\pm|
			\leq \|\ket{\tilde\omega^\pm_\star} - \ket{\omega^\pm_\star}\|_2
			= O\!\left(\frac{1}{\sqrt{n}}\right).
		\end{equation}
		Combining both parts for each $\pm$ term and summing:
		\begin{equation}
			\|(\tilde{W}^t - W^t)| \tilde{\psi}_0^{\parallel} \rangle\|_2
			= O\!\left(\frac{1}{\sqrt{n}}\right) + O\!\left(\frac{t}{n^2}\right).
		\end{equation}
		Therefore the spectral contribution is $O(1/\sqrt{n}) + O(t/n^2)$.
		
		\medskip
		\noindent\textbf{\textit{Merging the contributions.}}
		Collecting all contributions:
		\begin{equation}
			|\tilde{p}_M(t) - p_M(t)|
			= O\!\left(\frac{1}{\sqrt{n}}\right)
			+ O\!\left(\frac{t}{n^2}\right),
		\end{equation}
		and for $t = t^\star = O(1)$:
		\begin{equation}
			|\tilde{p}_M(t^\star) - p_M(t^\star)|
			= O\!\left(\frac{1}{\sqrt{n}}\right).\qedhere
		\end{equation}
	\end{proof}
}
%SIXTH APPENDIX END---------------------------------------------------------------------------------------------------

%SEVENTH APPENDIX BEGIN-------------------------------------------------------------------------------------
\toappendix{%
	\stepcounter{AppCount}
	\subsection*{\Alph{AppCount}.\space Reformulating Graph Completeness Testing Algorithm with Analytical Foundations}
	\labelshort[\Alph{AppCount}]{app:algorithm_reformulation}
	
	In this appendix we revisit the graph completeness testing algorithm introduced in~\cite{GIORDANO2026170305} and reformulate its two stages on the rigorous analytical basis established in the present work. Two aspects of the original algorithm rested on numerical evidence:
	\begin{enumerate}
		\item[\textup{(i)}] The strict decrease of the gap eigenvalue $\lambda_\star$ under edge removal, which was stated as a conjecture extrapolated from simulations and by analogy with the spectral radius of the adjacency matrix.
		\item[\textup{(ii)}] The scaling of the number of qubits in the first register of the QPE procedure, which was derived from an empirical fit of simulated phase differences, with no analytical justification for the fitted exponent.
	\end{enumerate}
	Both these assumptions are now replaced by the rigorous results proved in the main text.
	
	\paragraph{Overview of the algorithm.} The quantum completeness testing algorithm decides whether an undirected input graph $G$ with $n$ nodes is complete. It uses the Szegedy quantum walk with marked nodes and the quantum phase estimation subroutine, and operates in two stages.
	
	\medskip
	\noindent\textit{Stage 1: Ruling-out.}
	The optimal number of marked nodes $m^* = \lfloor(n-1)/a\rfloor$, with $a = 1.44512$, is computed from the linear optimality condition $n = am^*+1$. The walk operator $W_{P'}$ is applied to the uniform initial state $|\psi_0\rangle$ for $t^* = 3$ steps. If the measured node lies outside the marked set $M$, the graph is declared incomplete and the algorithm halts. Otherwise, Stage 2 is entered. The correctness of this ruling-out step rests on two facts: for a complete graph the success probability $P^*_M(t^*,n) \approx 1$, thus if the graph is complete we should find the walker on a node in the marked set $M$, hence not finding a walker on the marked node means that the graph is surely not complete. Notice that, on the contrary, finding a walker on a marked node does not necessarily imply completeness.
	
	\medskip
	\noindent\textit{Stage 2: Verification via QPE.}
	If the first stage fails to give an answer, i.e. we find a marked node, we proceed to this verification second stage. We change the number of marked nodes to $1$ ($m = 1$), the expected eigenphase of the complete graph is,
	\begin{equation}
		\theta_2 = \arccos\!\left(\frac{n-2}{n-1}\right),
	\end{equation}
	and is computed analytically. Also the eigenstate $|\theta^+_2\rangle$ of the complete graph is known by \cite{Portugal_book}. QPE is applied to the walk operator $W_{P''}$ with the eigenstate $|\theta^+_2\rangle$ in the second register. The estimated phase $\theta_j$ is compared with $\theta_2$: equality confirms completeness, any discrepancy confirms incompleteness.
	
	\paragraph{Analytical basis for Stage 1: formal proof of the eigenvalue decrease.}
	
	In~\cite{GIORDANO2026170305} the inequality $\tilde\lambda_\star < \lambda_\star$ for an incomplete graph was supported by simulations and assumed by analogy with the adjacency matrix spectral radius, without a proof at the level of the transition matrix. \Cref{prop:negative_shift} of the present work provides this missing proof in case the incompleteness concerns the unmarked sector: when a single edge $\{u,v\} \subseteq V \setminus M$ is removed from $K_n$, the perturbed gap eigenvalue satisfies $\tilde\lambda_\star<\lambda_\star$ strictly, for every $n\ge3$ and every $0\le m\le n-2$. \Cref{prop:exact_shift,prop:asymptotic_shift} further quantify the shift exactly and asymptotically:
	\begin{equation}\label{eq:app_shift_stage1}
		\Delta\lambda_\star
		= -\,\frac{2m}{(n-1)(n-2)(n-m)} - \Theta\!\left(\frac{1}{n^3}\right)
		= -\,\Theta\!\left(\frac{1}{n^2}\right),
	\end{equation}
	in the intermediate regime $m = \Theta(n)$ used in Stage 1 (\Cref{prop:asymptotic_shift}). Since $\cos(\tilde\theta_\star) = \tilde\lambda_\star < \lambda_\star = \cos(\theta_\star)$ and $\arccos$ is strictly decreasing, the phase
	$\tilde\theta_\star > \theta_\star$ and the walk dynamics on the incomplete graph differ from the complete case, lowering the success probability. This provides the rigorous mechanism underlying the ruling-out step.
	
	\paragraph{Analytical basis for Stage 2: qubit count from asymptotic phase scaling.}
	
	In~\cite{GIORDANO2026170305}, the minimum resolvable phase difference $\theta_j - \theta_2$ was studied numerically for graphs with $n = 4$ to $n = 300$ nodes, each with one edge removed and $m=1$ nodes marked. To set a lower bound for this difference, a least-square fit of the simulated data was done and its result adjusted getting the following function:	
	\begin{equation}\label{eq:empirical_bound}
		\mathcal{F}(n) \sim \frac{13}{n^{3.4}}.
	\end{equation}
	From $\mathcal{F}(n)$ the QPE qubit count was set as $p \sim |3.4\log_2(13n)| + 1 = O(\log_2 n)$. We stress that $\mathcal F(n)$ was deliberately chosen as a conservative lower bound on the simulated phase differences and on the fit of the data, as it can be seen from \cite{GIORDANO2026170305}, Fig. C7: the direct result of the nonlinear fit gives a decay closer to $n^{-2.7}$, and $\mathcal F(n)$ was tuned to sit safely below it, to yield a lower bound for the phase difference and an upper bound on the required qubit count and the QPE complexity. No analytical derivation of any of these exponents was available. The present study provides it.
	
	Stage 2 uses $m = 1$ marked node and the worst-case input is a graph with exactly one edge removed, precisely the case analyzed in the present work. Setting $m=1$ in \Cref{prop:phase_shift_explicit}, we are exactly in the few-marked-nodes regime with $m=O(1)$, which gives
	\begin{equation}
		\Delta\theta_\star \;=\; \Theta\!\left(n^{-5/2}\right).
	\end{equation}
	
	This analytical $\Theta(n^{-5/2})$ decay is the exact asymptotic scaling of the phase difference that Stage 2 must resolve. It is consistent with, and indeed lies closer to, the empirical observation in~\cite{GIORDANO2026170305}: the nonlinear fit decay rate ($\sim n^{-2.7}$) is markedly nearer to the analytical $n^{-5/2}=n^{-2.5}$ derived here than to the conservative empirical bound $\mathcal F(n)\sim n^{-3.4}$ used to set the qubit count in that work. The bound $\mathcal F(n)$ remains valid since $\mathcal F(n)\sim n^{-3.4}$ decays faster than $n^{-5/2}$ for all $n$ larger than a small constant. $\mathcal F(n)$ is still a valid conservative lower bound on the actual phase differences for all finite $n$ in the simulated range, as required by the construction of~\cite{GIORDANO2026170305}.
	
	Since $\Delta\theta_\star = \Theta(n^{-5/2})$ is the minimum phase difference to be resolved, the number of qubits in the QPE first register must satisfy
	\begin{equation}\label{eq:app_qubit_count}
		p \;\sim\; \left|\log_2\!\left(\Delta\theta_\star\right)\right| + 1
		\;\sim\; \frac{5}{2}\,\log_2 n + O(1),
	\end{equation}
	replacing the empirical estimate of $\sim 3.4\log_2(13n) + 1$. Both yield the same asymptotic complexity class, since the inverse quantum Fourier transform requires $\mathcal{O}(p^2)$ steps, giving QPE complexity $=\mathcal{O}\!\left(\log^2 n\right)$.
	
	The total time complexity of the completeness testing algorithm therefore remains $\mathcal{O}(\log^2 n)$, now established on a purely analytical basis: the constant running time $t^* = 3$ of Stage 1 is negligible for $n \to \infty$, and the QPE qubit count in Stage 2 scales as $\frac{5}{2}\log_2 n$, derived directly from the $\Theta(n^{-5/2})$ phase shift quantified in \Cref{prop:phase_shift_explicit}.
}

%SEVENTH APPENDIX END--------------------------------------------------------------------------------------------------------------------------------

	\section{Conclusions}
	Grounding on the completeness testing algorithm formulation \cite{GIORDANO2026170305} we have presented a rigorous perturbation analysis of the Szegedy quantum walk search algorithm on the complete graph $K_n$ with marked nodes, specifically addressing the structural defect caused by the removal of a single undirected edge. Our analysis demonstrates that this minimal topological change induces a perturbation to the transition matrix with a spectral norm scaling as $O(1/n)$. By examining the block structure of the discriminant matrix, we proved that this perturbation leads to a systematic negative shift of the spectral gap eigenvalue $\lambda_\star$, for every $n$ and every number of marked nodes $m$; in the regime relevant for the search algorithm, where the number of marked nodes $m$ is proportional to $n$, this shift has magnitude $\Theta(1/n^2)$. Consequently, the corresponding phase shift $\Delta\theta_\star$, which governs the frequency of the quantum search, also scales as $\Theta(1/n^2)$ in this regime.
	
	Furthermore, we evaluated the impact on the success probability $p_M(t)$, in the regime $m=\Theta(n)$. We established that the probability loss is upper-bounded by $O(1/\sqrt{n})$. Crucially, our derivation reveals that this loss is primarily driven by the geometric misalignment of the effective subspace $\mathcal{V}_{\mathrm{eff}}$ and its perturbed counterpart $\tilde{\mathcal{V}}_{\mathrm{eff}}$, controlled by
	$\|\Pi_{\tilde{\mathcal{V}}_{\mathrm{eff}}} - \Pi_{\mathcal{V}_{\mathrm{eff}}}\|_2 = O(1/\sqrt{n})$, rather than by the shift in the search frequency or leakage into the bulk spectrum. These results, together with those of our seminal work \cite{GIORDANO2026170305}, compose an analytical base in the framework of quantum walk based graph properties evaluation. First, graph-completeness testing provides the baseline quantum decision procedure for determining whether a graph or sub-graph is complete. Second, the single-edge-removal perturbative analysis refines this baseline by quantifying the response of the procedure to the minimal nontrivial anomaly. Together, these two components define a coherent approach for quantum-assisted topology-integrity monitoring, and are unified in \nameshortref{app:algorithm_reformulation}, where the completeness testing algorithm of~\cite{GIORDANO2026170305} is reformulated replacing its numerical foundations with the analytical results of the present work.
	
	In the application scenario outlined in \Cref{sec:motivation_and_application}, these results have a specific interpretation. The $O(1/\sqrt{n})$ bound on the probability loss quantifies the resilience of the quantum monitoring procedure: in a large trusted network with $n$ entities, the impact of a single link failure on the walk dynamics is provably small but detectable, since the spectral shift $\Delta\theta_\star = \Theta(1/n^2)$ is nonzero and analytically characterized in the regime $m=\Theta(n)$. The dominant mechanism is geometric, consisting in the misalignment of the effective subspace $\mathcal{V}_{\mathrm{eff}}$, while the change in search frequency is less relevant, which suggests that monitoring the subspaces can be more effective than eigenphase estimation.
	
	The framework therefore provides a foundation for quantum-assisted anomaly detection in environments such as control-plane overlays or data-center fabrics, where the expected topology is complete and even a single missing link is operationally significant. In such settings, our analysis establishes both the theoretical bounds and the limitations of the approach: the method is sensitive to the minimal	structural anomaly, but the effect scales as $O(1/\sqrt{n})$ for the probability of finding a marked node, which means that as the network grows, progressively higher precision is required for reliable detection.
	
	The next step would be to extend the analysis from one missing edge to multiple simultaneous edge failures, edges with non-uniform failure probabilities, or connectivity restrictions imposed by external policies. Another natural direction is to complement the analytical perturbation theory with numerical simulations on finite-size network instances motivated by cybersecurity use cases. Such developments would help clarify the practical detection thresholds of the method and the regimes in which the spectral response remains sufficiently strong for robust anomaly identification.

	\section*{Acknowledgements}
	The authors acknowledge support from Spanish MICIN grant PID2021-122547NB-I00 and the ``MADQuantum-CM'' project funded by Comunidad de Madrid (Programa de Acciones Complementarias) and by the Ministry for Digital Transformation and of Civil Service of the Spanish Government through the QUANTUM ENIA project call -- Quantum Spain project, and by the European Union through the Recovery, Transformation and Resilience Plan Next Generation EU within the framework of the Digital Spain 2026 Agenda, the CAM Programa TEC-2024/COM-84 QUITEMAD-CM. This work has been financially supported by the project MADQuantum-CM, funded by Comunidad de Madrid (Programa de Acciones Complementarias) and by the Recovery, Transformation and Resilience Plan -- Funded by the European Union -- (NextGeneration EU, PRTR-C17.I1). SG acknowledges support from the project MADQuantum-CM, funded by Comunidad de Madrid (Programa de Acciones Complementarias) and by the Recovery, Transformation and Resilience Plan -- Funded by the European Union -- (NextGeneration EU, PRTR-C17.I1). M.A. M.-D. has been partially supported by the U.S. Army Research Office through Grant No. W911NF-14-1-0103. We also acknowledge the ELLIS Unit Madrid.
	\end{multicols*}
	
	\appendix
	\section*{Appendices}
		\printappendix
		\bibliography{bibliography.bib}
		\newpage
		\section*{Variables}
		\begin{table}[h!]
			\centering
			\renewcommand{\arraystretch}{1.3}
			\begin{tabular}{|c|c|c|}
				\hline
				\textbf{Symbol} & \textbf{Description} & \textbf{Space / Range} \\
				\hline\hline
				\multicolumn{3}{|c|}{\textbf{Graph and Transition Matrices}} \\
				\hline
				$V$ & Set of nodes & -- \\
				$\mathcal{E}^0$ & Set of edges of the baseline graph & -- \\
				$G^0=(V,\mathcal{E}^0)$ & Baseline graph with $n$ nodes & -- \\
				$n$ & Number of nodes & $\in \mathbb{N}$ \\
				$M$ & Set of marked nodes & $M \subseteq V$ \\
				$m = |M|$ & Number of marked nodes & $\in \{1,\ldots,n-2\}$ \\
				$U = V \setminus M$ & Set of unmarked nodes & $U \subseteq V$ \\
				$k = n - m$ & Number of unmarked nodes & $\in \mathbb{N}$ \\
				$u, v$ & Endpoints of the removed edge & $\in U$ \\
				$P^0$ & Column-stochastic transition matrix, no marking & $\mathbb{R}^{n \times n}$ \\
				$P$ & Transition matrix with marked (absorbing) nodes & $\mathbb{R}^{n \times n}$ \\
				\hline
				\multicolumn{3}{|c|}{\textbf{Matrix Columns and Basis Vectors}} \\
				\hline
				$\mathbf{e}_j$ & Standard basis vector ($1$ at index $j$, $0$ elsewhere) & $\in \mathbb{R}^n$ \\
				$P_{\cdot,j}$ & $j$-th column of the transition matrix & $\in \mathbb{R}^n$ \\
				$\mathbf{c}$ & Original column of node $u$, i.e.\ $P_{\cdot,u}$ & $\in \mathbb{R}^n$ \\
				$\tilde{\mathbf{c}}$ & Perturbed column of node $u$ after edge removal & $\in \mathbb{R}^n$ \\
				$\mathbf{1}_U$ & All-ones vector supported on $U$ & $\in \mathbb{R}^n$ \\
				\hline
				\multicolumn{3}{|c|}{\textbf{Perturbation Analysis}} \\
				\hline
				$p = \frac{1}{n-1}$ & Transition probability weight of the removed edge & $\in (0,1)$ \\
				$\mathbf{r}^{(u)}, \mathbf{r}^{(v)}$ & Redistribution vectors for nodes $u$ and $v$, uniform over $V\setminus\{u,v\}$ & $\in \mathbb{R}^n$ \\
				$q$ & Redistributed weight received by each unmarked target, $q=p\bigl(\sqrt{1+\tfrac1{n-2}}-1\bigr)$ & $\in \mathbb{R}_{>0}$ \\
				$d$ & $d:=\sqrt{n-1}-\sqrt{n-2}$ & $\in \mathbb{R}_{>0}$ \\
				$\mathbf{a} = \mathbf{r}^{(u)}-\mathbf{e}_v$ & Auxiliary perturbation vector for column $u$ & $\in \mathbb{R}^n$ \\
				$\mathbf{b} = \mathbf{r}^{(v)}-\mathbf{e}_u$ & Auxiliary perturbation vector for column $v$ & $\in \mathbb{R}^n$ \\
				$\Delta$ & Perturbation of the transition matrix, rank $2$ & $\mathbb{R}^{n \times n}$ \\
				$E = \tilde{C}-C$ & Perturbation of the discriminant matrix, symmetric & $\mathbb{R}^{n \times n}$ \\
				$\|\cdot\|_F$ & Frobenius norm & $\to \mathbb{R}_{\geq 0}$ \\
				$\|\cdot\|_2$ & Spectral norm (largest singular value) & $\to \mathbb{R}_{\geq 0}$ \\
				$\delta_\star,\, \delta_B$ & Spectral gaps of $C$ used in Davis--Kahan bounds & $\in \mathbb{R}_{>0}$ \\
				$\delta_W$ & Spectral gap of $W$ separating $\mathcal{V}_{\star}$ from the rest of the spectrum & $\in \mathbb{R}_{>0}$ \\
				$\Delta\lambda_\star$ & First-order shift of the gap eigenvalue & $\in \mathbb{R}_{<0}$ \\
				$T_A,\,T_B$ & Nonnegative decomposition $|\Delta\lambda_\star|=T_A+T_B$ & $\in \mathbb{R}_{\ge0}$ \\
				$\Delta\theta_\star$ & First-order shift of the corresponding eigenphase & $\in \mathbb{R}_{>0}$ \\
				$S_A,\,S_B$ & Nonnegative decomposition $\Delta\theta_\star=S_A+S_B+O(1/n^2)$ & $\in \mathbb{R}_{\ge0}$ \\
				$\alpha_{\mathrm{eff}}$ & Largest principal angle between
				$\mathcal{U}_{\mathrm{eff}}$ and $\tilde{\mathcal{U}}_{\mathrm{eff}}$
				in $\mathbb{R}^n$ & $\in [0,\pi/2]$ \\
				\hline
			\end{tabular}
		\end{table}
		
		\begin{table}[h!]
			\centering
			\renewcommand{\arraystretch}{1.3}
			\begin{tabular}{|c|c|c|}
				\hline
				\textbf{Symbol} & \textbf{Description} & \textbf{Space / Range} \\
				\hline\hline
				\multicolumn{3}{|c|}{\textbf{Szegedy Quantum Walk Operators and Vectors}} \\
				\hline
				$\ket{u}, \ket{v}$ & Computational basis states representing nodes & $\in \mathbb{C}^n$ \\
				$\mathcal{H} = \mathbb{C}^n \otimes \mathbb{C}^n$ & Hilbert space of the walk,
				$\mathrm{span}\{\ket{x}\ket{y}\}$ & -- \\
				$\ket{\psi_u},\, \ket{\phi_v}$ & States in $\mathcal{H}$ used to define
				the reflections & $\in \mathcal{H}$ \\
				$A = \sum_u|\psi_u\rangle\langle u|$ & Isometry mapping $\mathbb{R}^n$
				into $\mathcal{H}$, $\|A\|_2=1$ & $\mathbb{R}^n \to \mathcal{H}$ \\
				$B = \sum_v|\phi_v\rangle\langle v|$ & Isometry mapping $\mathbb{R}^n$
				into $\mathcal{H}$, $\|B\|_2=1$ & $\mathbb{R}^n \to \mathcal{H}$ \\
				$\tilde{A},\, \tilde{B}$ & Perturbed counterparts of $A$ and $B$ & $\mathbb{R}^n \to \mathcal{H}$ \\
				$\ket{\psi_0}$ & Initial state:
				$\frac{1}{\sqrt{n}}\sum_{x \in V}\ket{\psi_x}$ & $\in \mathcal{H}$ \\
				$|\psi_0^{\parallel} \rangle = \Pi_{\mathcal{V}_{\mathrm{eff}}}\ket{\psi_0}$ &
				Component of $\ket{\psi_0}$ in $\mathcal{V}_{\mathrm{eff}}$ & $\in \mathcal{H}$ \\
				$| \tilde{\psi}_0^{\parallel} \rangle = \Pi_{\tilde{\mathcal{V}}_{\mathrm{eff}}}\ket{\psi_0}$ &
				Component of $\ket{\psi_0}$ in $\tilde{\mathcal{V}}_{\mathrm{eff}}$ & $\in \mathcal{H}$ \\
				$\ket{\psi_0^\perp} = \Pi_{\mathcal{A}^\perp \cap \mathcal{B}^\perp}\ket{\psi_0}$ &
				Component of $\ket{\psi_0}$ in $\mathcal{A}^\perp \cap \mathcal{B}^\perp$ & $\in \mathcal{H}$ \\
				$p_M^\parallel(t) = \|\Pi_M W^t|\psi_0^{\parallel} \rangle\|^2$ &
				Unperturbed success probability restricted to $\mathcal{V}_{\mathrm{eff}}$ & $\in [0,1]$ \\
				$\tilde{p}_M^\parallel(t) = \|\Pi_M\tilde{W}^t| \tilde{\psi}_0^{\parallel} \rangle\|^2$ &
				Perturbed success probability restricted to $\tilde{\mathcal{V}}_{\mathrm{eff}}$ & $\in [0,1]$ \\
				$R_A,\, R_B$ & Reflection operators (unperturbed) & Unitary on $\mathcal{H}$ \\
				$W = R_B R_A$ & Szegedy walk operator (unperturbed) & Unitary on $\mathcal{H}$ \\
				$\tilde{W}$ & Szegedy walk operator (perturbed) & Unitary on $\mathcal{H}$ \\
				$C,\, \tilde{C}$ & Discriminant matrix (original and perturbed)
				& $\mathbb{R}^{n \times n}$, symmetric \\
				$\Pi_M$ & Orthogonal projector onto the marked subspace & -- \\
				$p_M(t) = \|\Pi_M W^t \ket{\psi_0}\|^2$ & Success probability at time $t$
				& $\in [0,1]$\\
				\hline
	\end{tabular}
\end{table}

\begin{table}[h!]
\centering
\renewcommand{\arraystretch}{1.3}
\begin{tabular}{|c|c|c|}
\hline
\textbf{Symbol} & \textbf{Description} & \textbf{Space / Range} \\
\hline\hline
				\multicolumn{3}{|c|}{\textbf{Spectral Quantities and Subspaces}} \\
				\hline
				$\lambda_M = 1$ & Eigenvalue of $C$ assoc.\ with $\mathcal{U}_M$,
				mult.\ $m$ & $\in [-1,1]$ \\
				$\lambda_\star = \frac{n-m-1}{n-1}$ & Simple gap eigenvalue of $C$
				& $\in (-1,1)$ \\
				$\lambda_B = -\frac{1}{n-1}$ & Bulk eigenvalue of $C$, multiplicity $k-1$
				& $\in [-1,1]$ \\
				$\theta_M = 0,\, \theta_\star,\, \theta_B$ & Eigenphases of $W$,
				defined by $\cos\theta = \lambda$ & $\in [0,\pi]$ \\
				$\mathbf{v}_\star = \frac{1}{\sqrt{k}}\mathbf{1}_U$ &
				Normalized eigenvector of $\lambda_\star$ & $\in \mathbb{R}^n$, $\|\mathbf{v}_\star\|=1$ \\
				$\mathcal{U}_M$ & Eigenspace of $C$ for $\lambda_M=1$,
				subspace of $\mathbb{R}^n$ & Dim.\ $m$ \\
				$\mathcal{U}_\star$ & Eigenspace of $C$ for $\lambda_\star$,
				subspace of $\mathbb{R}^n$ & Dim.\ $1$ \\
				$\mathcal{U}_B$ & Bulk eigenspace of $C$ for $\lambda_B$,
				subspace of $\mathbb{R}^n$ & Dim.\ $k-1$ \\
				$\mathcal{A} = \mathrm{ran}(A)$ & Subspace of $\mathcal{H}$ spanned
				by $\{|\psi_u\rangle : u \in V\}$ & $\subset \mathcal{H}$ \\
				$\mathcal{B} = \mathrm{ran}(B)$ & Subspace of $\mathcal{H}$ spanned
				by $\{|\phi_v\rangle : v \in V\}$ & $\subset \mathcal{H}$ \\
				$\mathcal{A}^\perp \cap \mathcal{B}^\perp$ & $(+1)$-eigenspace of $W$
				orthogonal to $\mathcal{V}_{\mathrm{eff}} \oplus \mathcal{V}_B$ &
				Dim.\ $n^2-2n+m$ \\
				$\mathcal{V}_M$ & Invariant subspace of $W$ in $\mathcal{H}$
				corresponding to $\mathcal{U}_M$ & Dim.\ $m$ \\
				$\mathcal{V}_\star$ & Invariant subspace of $W$ in $\mathcal{H}$
				corresponding to $\mathcal{U}_\star$ & Dim.\ $2$ \\
				$\mathcal{V}_B$ & Bulk invariant subspace of $W$ in $\mathcal{H}$
				corresponding to $\mathcal{U}_B$ & Dim.\ $2(k-1)$ \\
				$\mathcal{V}_{\mathrm{eff}} = \mathcal{V}_M \oplus \mathcal{V}_\star$ &
				Effective invariant subspace of $W$ in $\mathcal{H}$ & Dim.\ $m+2$ \\
				$\tilde{\mathcal{U}}_{\mathrm{eff}},\,\tilde{\mathcal{U}}_M,\,
				\tilde{\mathcal{U}}_\star,\,\tilde{\mathcal{U}}_B$ &
				Perturbed counterparts of eigenspaces of $\tilde{C}$ in $\mathbb{R}^n$
				& $\subset \mathbb{R}^n$ \\
				$\tilde{\mathcal{V}}_{\mathrm{eff}},\,\tilde{\mathcal{V}}_M,\,
				\tilde{\mathcal{V}}_\star,\,\tilde{\mathcal{V}}_B$ &
				Perturbed counterparts of invariant subspaces of $\tilde{W}$
				& $\subset \mathcal{H}$ \\
				$\ket{\omega^\pm_\star}$ & Eigenvectors of $W$ in $\mathcal{V}_\star$
				for $e^{\pm 2i\theta_\star}$ & $\in \mathcal{H}$ \\
				$\ket{\omega_{M,j}}$ & Eigenvectors of $W$ in $\mathcal{V}_M$
				for eigenvalue $1$, $j=1,\ldots,m$ & $\in \mathcal{H}$ \\
				$\ket{\omega_{B,l}}$ & Eigenvectors of $W$ in $\mathcal{V}_B$
				& $\in \mathcal{H}$ \\
				$\ket{\tilde{\omega}^\pm_\star}$ & Eigenvectors of $\tilde{W}$
				in $\tilde{\mathcal{V}}_\star$ for $e^{\pm 2i\tilde\theta_\star}$
				& $\in \mathcal{H}$ \\
				$\ket{\tilde{\omega}_{M,j}}$ & Eigenvectors of $\tilde{W}$
				in $\tilde{\mathcal{V}}_M$ for eigenvalue $1$ & $\in \mathcal{H}$ \\
				$\ket{\tilde{\omega}_{B,l}}$ & Eigenvectors of $\tilde{W}$
				in $\tilde{\mathcal{V}}_B$ & $\in \mathcal{H}$ \\
				$a^\pm$ & Coefficients of $| \tilde{\psi}_0^{\parallel} \rangle$
				in eigenbasis of $W$ & $\in \mathbb{C}$ \\
				$\tilde{a}^\pm$ & Coefficients of $| \tilde{\psi}_0^{\parallel} \rangle$
				in eigenbasis of $\tilde{W}$ & $\in \mathbb{C}$ \\
				$c_j = \tilde{c}_j$ & Coefficients of $| \tilde{\psi}_0^{\parallel} \rangle$
				for $\mathcal{V}_M$ components (equal by block structure) & $\in \mathbb{C}$ \\
				\hline
			\end{tabular}
			\caption{Summary of notation used throughout the document.}
			\label{tab:notation}
		\end{table}
	
	\end{document}